\documentclass[
    a4paper,
    12pt,
    DIV=11,
    headinclude=false,
    footinclude=false,
]{scrartcl}

\setkomafont{disposition}{\normalcolor\bfseries}

\newcommand{\LTadd}[1]{}

\newcommand{\LTinput}[1]{}
\newcommand{\YYCleverefInput}[1]{}

\usepackage[T1]{fontenc}
\usepackage[utf8]{inputenc}
\usepackage[main=english]{babel}
\usepackage{csquotes}
\usepackage{microtype}

\usepackage{mathtools}
\mathtoolsset{centercolon,mathic}
\allowdisplaybreaks[2] 
\usepackage{amssymb}
\ifdefined\newcounteralias
    \RemoveFromHook{package/amsthm/after}[firstaid/aliascounter]
\fi
\usepackage{amsthm}
\usepackage{thmtools}
\usepackage{derivative}
\derivset{\odv}[sep-frac-fun=\,]
\usepackage{enumitem}
\usepackage{booktabs}
\usepackage{makecell}
\usepackage[section]{placeins}

\newcommand{\statement}[1]{\paragraph{#1}\pdfbookmark[1]{#1}{#1}} 

\renewcommand{\Delta}{\varDelta}

\renewcommand{\Phi}{\varPhi}

\renewcommand{\Psi}{\varPsi}

\renewcommand{\Lambda}{\varLambda}

\renewcommand{\Gamma}{\varGamma}

\renewcommand{\Omega}{\varOmega}

\let\subset\subseteq

\let\epsi\varepsilon

\DeclarePairedDelimiter{\intervaloo}{\lparen}{\rparen}

\DeclarePairedDelimiter{\intervalco}{\lbrack}{\rparen}
\DeclarePairedDelimiter{\intervalcc}{\lbrack}{\rbrack}

\DeclarePairedDelimiter{\paren}{\lparen}{\rparen}
\DeclarePairedDelimiterXPP{\pospart}[1]{}{\lbrack}{\rbrack}{_+}{#1}
\DeclarePairedDelimiter{\abs}{\lvert}{\rvert}
\DeclarePairedDelimiter{\norm}{\lVert}{\rVert}
\DeclarePairedDelimiterXPP{\cbn}[1]{}{\lVert}{\rVert}{_{\mathup{cb}}}{#1}

    \newcommand{\VERT}[1]{#1|\mkern-1.5mu#1|\mkern-1.5mu#1|}

    \ExplSyntaxOn
    \makeatletter

    \MT_delim_default_inner_wrappers:n{\tnorm}

    \NewDocumentCommand{\tnorm}{ s o m }{
        \IfBooleanTF{#1}{
        \MT_delim_tnorm_star_wrapper:nnn%
            {\VERT{\bgroup\left}}{#3}{\VERT{\aftergroup\egroup\right}}
        }{
            \IfValueTF{#2}{
                \@nameuse{MT_delim_tnorm_nostarscaled_wrapper:nnn}%
                    {\VERT{\@nameuse {\MH_cs_to_str:N #2 l}}}
                    {#3}
                    {\VERT{\@nameuse {\MH_cs_to_str:N #2 r}}}
            }{
                \MT_delim_tnorm_nostarnonscaled_wrapper:nnn%
                    {\VERT{}}
                    {#3}
                    {\VERT{}}
            }
        }
    }

    \makeatother
    \ExplSyntaxOff

\DeclarePairedDelimiter{\commutator}{\lbrack}{\rbrack}

\DeclarePairedDelimiter{\acomm}{\lbrace}{\rbrace}

\DeclarePairedDelimiterX{\innerproduct}[2]{\langle}{\rangle}{#1,#2}
\DeclarePairedDelimiterX{\innerp}[2]{\langle}{\rangle}{#1,#2}

\DeclarePairedDelimiter{\List}{\{}{\}}
\DeclarePairedDelimiter{\ceil}{\lceil}{\rceil}

\DeclarePairedDelimiterXPP{\dd}[1]{d}{\lparen}{\rparen}{}{#1}
\DeclarePairedDelimiterXPP{\dist}[1]{d}{\lparen}{\rparen}{}{#1}
\DeclarePairedDelimiterXPP{\diam}[1]{\SYMdiam}{\lparen}{\rparen}{}{#1}
\DeclarePairedDelimiterXPP{\Exp}[1]{\exp}{\lparen}{\rparen}{}{#1}

\makeatletter
\let\bBigg@@\bBigg@
\renewcommand{\bBigg@}[2]{{%
  \mathchoice
    {\bBigg@@{#1}{#2}}%
    {\bBigg@@{#1}{#2}}%
    {\big@size=.5\big@size\bBigg@@{#1}{#2}}%
    {\big@size=.3\big@size\bBigg@@{#1}{#2}}}}%
\makeatother

\DeclareDocumentCommand{\trace}{s o e{_} m}{
    \IfValueTF{#3}
        {\Tr_{#3}}
        {\Tr}%
    \IfBooleanTF{#1}
        {\paren*{#4}}
        {
            \IfValueTF{#2}
                {\paren[#2]{#4}}
                {\paren{#4}}%
        }%
}
\DeclareDocumentCommand{\Trace}{s o e{_} m}{
    \IfValueTF{#3}
        {\Tr_{#3}}
        {\Tr}%
    \IfBooleanTF{#1}
        {\paren*{#4}}
        {
            \IfValueTF{#2}
                {\paren[#2]{#4}}
                {\paren{#4}}%
        }%
}

\providecommand\given{}

\newcommand\SetSymbol[1][]{%
    \nonscript\,#1\vert
    \allowbreak
    \nonscript\,
    \mathopen{}}
\DeclarePairedDelimiterX\Set[1]\{\}{%
    \renewcommand\given{%
        \SetSymbol[\delimsize]}
    \nonscript\,
    #1
    \nonscript\,
}

\ExplSyntaxOn
\makeatletter
\cs_generate_variant:Nn \tl_if_eq:nnT {onT}
\tl_new:N \__scale_new_delimsize_tl
\cs_new:Nn \__scale_make_bigger_tl:N {
    \tl_if_eq:onT {#1} {\@MHempty} { \tl_set:Nn \__scale_new_delimsize_tl { big } }
    \tl_if_eq:onT {#1} {}          { \tl_set:Nn \__scale_new_delimsize_tl { big } }
    \tl_if_eq:onT {#1} {\big}      { \tl_set:Nn \__scale_new_delimsize_tl { Big } }
    \tl_if_eq:onT {#1} {\Big}      { \tl_set:Nn \__scale_new_delimsize_tl { bigg } }
    \tl_if_eq:onT {#1} {\bigg}     { \tl_set:Nn \__scale_new_delimsize_tl { Bigg } }
    \tl_if_eq:onT {#1} {\middle}   { \tl_set:Nn \__scale_new_delimsize_tl { middle } }
}
\cs_new:Nn \scale_make_bigger_m:N {
    \__scale_make_bigger_tl:N #1
    \use:c { \tl_use:N \__scale_new_delimsize_tl }
}
\cs_new:Nn \scale_make_bigger_l:N {
    \__scale_make_bigger_tl:N #1
    \tl_if_eq:NnTF \__scale_new_delimsize_tl {middle} {
        \tl_set:Nn \__scale_new_delimsize_tl {left}
    }{
        \tl_put_right:Nn \__scale_new_delimsize_tl {l}
    }
    \use:c { \tl_use:N \__scale_new_delimsize_tl }
}
\cs_new:Nn \scale_make_bigger_r:N {
    \__scale_make_bigger_tl:N #1
    \tl_if_eq:NnTF \__scale_new_delimsize_tl {middle} {
        \tl_set:Nn \__scale_new_delimsize_tl {right}
    }{
        \tl_put_right:Nn \__scale_new_delimsize_tl {r}
    }
    \use:c { \tl_use:N \__scale_new_delimsize_tl }
}
\DeclarePairedDelimiterXPP{\pdd}[1]{\scale_make_bigger_l:N\delimsize\lparen d}{\lparen}{\rparen}{\scale_make_bigger_r:N\delimsize\rparen}{#1}
\DeclarePairedDelimiterXPP{\pdist}[1]{\scale_make_bigger_l:N\delimsize\lparen d}{\lparen}{\rparen}{\scale_make_bigger_r:N\delimsize\rparen}{#1}
\DeclarePairedDelimiterXPP{\pdiam}[1]{\scale_make_bigger_l:N\delimsize\lparen \SYMdiam}{\lparen}{\rparen}{\scale_make_bigger_r:N\delimsize\rparen}{#1}

\makeatother
\ExplSyntaxOff

\usepackage{xcolor}
\definecolorset{HTML}{col_}{}{%
    1,0072B2;%
    2,D55E00;%
    3,E69F00%
}

\definecolor{linkcolor}{HTML}{3060c0}

\usepackage[
    colorlinks,
    allcolors=linkcolor,
    bookmarksnumbered,
    bookmarksopen,
    bookmarksopenlevel=2,
    bookmarksdepth=3,
    pdfusetitle=true,
    breaklinks=true,
]{hyperref}

\makeatletter
\@ifpackageloaded{hyperref}{
    \usepackage[
        capitalise,
        noabbrev,
        nameinlink,
    ]{cleveref}
    \usepackage{orcidlink}
}{
    \usepackage[
        capitalise,
        noabbrev,
        poorman,
    ]{cleveref}
    \newcommand{\texorpdfstring}[2]{#1}
    \newcommand{\href}[2]{#2}
    \newcommand{\hypersetup}[1]{}
    \newcommand{\orcidlink}[1]{ORCiD}
    \newcommand{\pdfbookmark}[1]{}
}
\makeatother

\crefdefaultlabelformat{#2#1#3}
\crefname{equation}{}{}
\crefname{assumption}{Assumption}{Assumptions}
\crefname{conjecture}{Conjecture}{Conjectures}

\usepackage{caption}
\usepackage{subcaption}

\newcommand{\sumstack}[2][]{\ifstrempty{#1}{\sum_{\substack{#2}}}{\smashoperator[#1]{\sum_{\substack{#2}}}}}

\NewDocumentCommand{\intlap}{E{_^}{{}{}}}{
    \int_{\mathrlap{#1}}^{\mathrlap{#2}}
}

\newcommand{\e}{{\mathrm{e}}}
\newcommand{\I}{\mathrm{i}}
 \newcommand{\R}{\mathbb{R}}
\newcommand{\C}{\mathbb{C}}
\newcommand{\N}{\mathbb{N}}
\newcommand{\Z}{\mathbb{Z}}
\newcommand{\HS}{{\mathcal{H}}}

\newcommand{\alg}{\mathcal{A}}

\DeclareMathOperator{\id}{id}
\newcommand\cexpsym{\mathbb{E}}
\ExplSyntaxOn
\DeclareDocumentCommand{\cexp}{s o m m}{%
    \cexpsym\c_math_subscript_token{#3}
    \IfBlankF{#4}
    {
        \exp_last_unbraced:Ne \paren {\IfBooleanT{#1}{*}\IfValueT{#2}{[\exp_not:N #2]}} {#4}
    }
}
\ExplSyntaxOff
\newcommand{\calB}{\mathcal{B}}

\newcommand{\calL}{\mathcal{L}}
\newcommand{\calQ}{\mathcal{Q}}
\newcommand{\calR}{\mathcal{R}}

\DeclareMathOperator{\Tr}{Tr}
\newcommand{\SYMdiam}{\operatorname{diam}}

\newcommand{\powerset}{\mathcal{P}_0}

\newcommand{\suchthat}{\mathpunct{\ordinarycolon}}

\newcommand{\quadtext}[1]{\quad\text{#1}\quad}
\newcommand{\qquadtext}[1]{\quad\quadtext{#1}\quad}

\newcommand{\Alignindent}{\hspace*{2em}&\hspace*{-2em}}

\newcommand{\mathup}[1]{\mathrm{#1}}

\newcommand{\colonpunct}{\mathpunct{:}}

\newcommand{\Seq}{\mathup{S}}

\newcommand{\Path}{\mathup{P}}
\newcommand{\SP}{\mathup{SP}}

\DeclareMathOperator{\scale}{scl}

\DeclareDocumentCommand{\F}{s o e{_} s o m m}{
    \def\decayFdist{
        \IfBooleanTF{#4}
            {\dist*{#6}}
            {
                \IfValueTF{#5}
                    {\dist[#5]{#6}}
                    {\dist{#6}}%
            }%
    }
    \def\decayFarg{
        \IfNoValueTF{#7}
            { \decayFdist }
            { \decayFdist - #7 }
    }%
    \IfValueTF{#3}
        { F_#3 }
        { F }%
    \IfBooleanTF{#1}
        {\paren*{\decayFarg}}
        {
            \IfValueTF{#2}
                {\paren[#2]{\decayFarg}}
                {\paren[\big]{\decayFarg}}%
        }%
}

\let\oldsetminus\setminus
\newbox\mybox
\newcommand\cutsetminus[1]{%
    \setbox\mybox\hbox{\(#1\oldsetminus\)}%
    \ht\mybox=0pt%
    \dp\mybox=0pt%
    \usebox\mybox%
}
\renewcommand\setminus{%
    \mathbin{%
        \mathchoice%
            {\displaystyle\oldsetminus}
            {\textstyle\oldsetminus}
            {\cutsetminus{\scriptstyle}}
            {\cutsetminus{\scriptscriptstyle}}
    }%
}

\makeatletter
\ExplSyntaxOn

\newcounter{theoremenv}
\newcounter{theoremenvglobal}
\newlist{thmlist}{enumerate}{1}
\setlist[thmlist]{
    label=\textup{(\roman{thmlisti})},
    ref={(\roman{thmlisti})},
    partopsep=0pt,
    parsep=0pt,
    itemsep=2pt,
    topsep=2pt,
}
\renewcommand{\p@thmlisti}{\perh@ps{\protect\ref{auto-label:\arabic{theoremenvglobal}}}}
\DeclareRobustCommand{\perh@ps}[1]{#1}
\newcommand{\itemref}[1]{%
    \begingroup 
    \let\perh@ps\@gobble\ref{#1}%
    \endgroup
}

\newcommand{\quiet@label}[1]{%
    \@ifundefined{@@label@@}{\label{#1}}{\@@label@@{#1}}%
}

\addtotheorempostheadhook{
    \crefalias{thmlisti}{\thmt@envname}
    \setcounter{theoremenv}{\value{\thmt@envname}}
    \stepcounter{theoremenvglobal}
    \exp_args:NNe \renewcommand \thetheoremenv {\use:c{the\thmt@envname}}
    \exp_args:Ne \quiet@label {auto-label:\arabic{theoremenvglobal}}
}
\makeatother
\ExplSyntaxOff

\newlength{\thmspacing}
\declaretheoremstyle[
    spaceabove=\thmspacing, spacebelow=\thmspacing,
    headfont=\normalfont\bfseries, bodyfont=\normalfont\itshape,
]{thmplain}

\declaretheoremstyle[
    spaceabove=\thmspacing, spacebelow=\thmspacing,
    headfont=\normalfont\bfseries, bodyfont=\normalfont,
]{thmdefinition}

\declaretheoremstyle[
    spaceabove=\thmspacing, spacebelow=\thmspacing,
    headfont=\normalfont\itshape, bodyfont=\normalfont,
]{thmremark}

\theoremstyle{thmplain}

\declaretheorem[
    name=Theorem,
    parent=section,
]{theorem}

\declaretheorem[
    name=Lemma,
    sibling=theorem,
]{lemma}

\declaretheorem[
    name=Proposition,
    sibling=theorem,
]{proposition}

\declaretheorem[
    name=Corollary,
    sibling=theorem,
]{corollary}

\theoremstyle{thmdefinition}

\declaretheorem[
    name=Definition,
    sibling=theorem,
]{definition}

\theoremstyle{thmremark}
\declaretheorem[
    name=Remark,
    sibling=theorem,
]{remark}

\makeatletter
\def\theHtheorem{\theHsection.\arabic{theorem}}
\@for\thm@name:=lemma,proposition,corollary,conjecture,definition,remark,example\do{%
    \expandafter\def\csname theH\thm@name\endcsname{\theHtheorem}%
}
\makeatother

\usepackage{tikz}
\newlength{\ul}
\newlength{\lw}
\pgfdeclarelayer{background}
\pgfdeclarelayer{lattice}
\pgfsetlayers{background,lattice,main}
\usetikzlibrary{
    arrows.meta,
}

\tikzset{%
    myLine/.style = {
        line width = #1\lw,
    },
    myLine/.default = 1,
    x = \ul,
    y = \ul,
    every layer/.style={
        myLine,
        text=black,
    },
    every picture/.style={
        x=\ul,
        y=\ul,
        every layer,
    },
    text = black,
    >={latex},
    shorten/.style = {
        shorten < = #1,
        shorten > = #1
    },
    shorten/.default = 2\lw,
    on layer/.code = {
        \pgfonlayer{#1}\begingroup
        \aftergroup\endpgfonlayer
        \aftergroup\endgroup
    },
    on background/.style = {
        preaction = {
            #1,
            on layer = background,
        },
    },
    only inside/.style = {
        preaction = {
            clip,
            postaction = {
                line width = 2\lw,
                #1
            },
        },
    },
    only inside/.default = {draw},
    set/.style = {
        shape = circle,
        myLine,
        minimum height = 1.5\ul,
        minimum width = 1.5\ul,
        preaction = {
            on layer = background,
            fill = #1!20,
            line width = \lw,
            draw = #1,
        },
        draw = #1,
        draw opacity = .3,
        text = #1,
        text opacity = 1,
        inner sep = 1\lw,
        outer sep = 0pt,
    },
    set/.default = gray,
    set label/.style = {
        shape = circle,
        minimum width = 0pt,
        minimum height = 0pt,
        on background = {},
        inner sep = 1\lw,
        outer sep = 2\lw,
        text = #1,
    },
    set label/.default = gray,
    point/.style = {
        circle,
        fill = #1,
        draw=none,
        inner sep=0,
        minimum width=2\lw,
        minimum height=2\lw,
        outer sep = 0pt,
    },
    point/.default = gray,
    big/.style = {
        minimum height = 2.5\ul,
        minimum width = 2.5\ul,
    },
    bounding note/.style = {
        use as bounding box,
        shape = rectangle,
        inner sep = 1\lw,
        outer sep = 0pt,
    },
    overprint/.style = {
        execute at begin node=\overprinton,
        execute at end node=\overprintoff,
    },
}

\usetikzlibrary{nfold}
\usetikzlibrary{spath3,intersections,nfold,decorations.pathreplacing}
\makeatletter
\newcommand\tikz@offset@apply[1]{%
    \pgfsyssoftpath@getcurrentpath\tikz@offset@keep
    \pgfsyssoftpath@setcurrentpath\pgfutil@empty
    \pgfoffsetpath{\tikz@actions@path}{#1}%
    \pgfsyssoftpath@getcurrentpath\tikz@actions@path
    \pgfsyssoftpath@setcurrentpath\tikz@offset@keep
}
\tikzset{
    offset/.code 2 args={%
        \tikzset{postaction={/utils/exec={\tikz@offset@apply{#1}},#2}}},
    offset/.default={5pt}{draw},
    grow region/.style 2 args = {
        line width={2*#1},
        line join=round,
        use path=#2,
    },
}

\newcommand\diffcircles[5]{%
  \path[name path=@out] (#1) circle (#2);
  \path[name path=@in]  (#3) circle (#4);
  \path[name intersections={of=@out and @in, by={@P,@Q}}];
  \pgfmathanglebetweenpoints{\pgfpointanchor{#1}{center}}{\pgfpointanchor{@P}{center}}\edef\myaP{\pgfmathresult}
  \pgfmathanglebetweenpoints{\pgfpointanchor{#1}{center}}{\pgfpointanchor{@Q}{center}}\edef\myaQ{\pgfmathresult}
  \pgfmathanglebetweenpoints{\pgfpointanchor{#3}{center}}{\pgfpointanchor{@P}{center}}\edef\mybP{\pgfmathresult}
  \pgfmathanglebetweenpoints{\pgfpointanchor{#3}{center}}{\pgfpointanchor{@Q}{center}}\edef\mybQ{\pgfmathresult}
  \pgfmathsetmacro\myRout{#2}\pgfmathsetmacro\myRin{#4}%
  \pgfmathsetmacro\myaQ{\myaQ<\myaP ? \myaQ+360 : \myaQ}%
  \path (#1) ++({(\myaP+\myaQ)/2}:#2) coordinate (@M);
  \pgfpointdiff{\pgfpointanchor{#3}{center}}{\pgfpointanchor{@M}{center}}%
  \pgfmathsetmacro\mydist{veclen(\the\pgf@x,\the\pgf@y)}%
  \ifdim\mydist pt<\myRin pt\relax\pgfmathsetmacro\myaQ{\myaQ-360}\fi
  \pgfmathsetmacro\mybP{\mybP<\mybQ ? \mybP+360 : \mybP}%
  \path (#3) ++({(\mybP+\mybQ)/2}:#4) coordinate (@M);
  \pgfpointdiff{\pgfpointanchor{#1}{center}}{\pgfpointanchor{@M}{center}}%
  \pgfmathsetmacro\mydist{veclen(\the\pgf@x,\the\pgf@y)}%
  \ifdim\mydist pt>\myRout pt\relax\pgfmathsetmacro\mybP{\mybP-360}\fi
  \path[save path=#5] (@P)
    arc[start angle=\myaP, end angle=\myaQ, radius=#2]
    arc[start angle=\mybQ, end angle=\mybP, radius=#4] -- cycle;
}
\makeatother

\usepackage{xurl}

\usepackage[
    bibstyle=phys,
    citestyle=numeric,
    biblabel=brackets,
    mincitenames=1,
    maxcitenames=2,
    minalphanames=3,
    maxalphanames=4,
    minbibnames=3,
    maxbibnames=10,
    arxiv=abs,
    eprint=true,
    doi=false,
    url=false,
    sorting=nty,
    useprefix=true,
    alldates=year,
    urldate=long,
]{biblatex}
\renewbibmacro*{note+pages}{%
    \printfield{note}%
    \setunit{\bibpagespunct}%
    \printfield{pages}%
    \newunit%
}
\makeatletter
\DeclareFieldFormat{eprint:arxiv}{%
    \ifhyperref
    {\href{https://arxiv.org/\abx@arxivpath/#1}{%
            arXiv\addcolon\linebreak[2]#1}}
    {arXiv\addcolon
        \nolinkurl{#1}%
    }
}
\makeatother

\renewbibmacro*{maintitle+title}{%
    \iffieldsequal{maintitle}{title}{%
        \clearfield{maintitle}%
        \clearfield{mainsubtitle}%
        \clearfield{maintitleaddon}%
    }{%
        \iffieldundef{maintitle}{}{%
            \printtext[doi/url-link]{%
                \usebibmacro{maintitle}%
            }%
            \newunit\newblock
            \iffieldundef{volume}{}{%
                \printfield{volume}%
                \printfield{part}%
                \setunit{\addcolon\space}%
            }%
        }%
    }%
    \printtext[doi/url-link]{%
        \usebibmacro{title}%
    }%
    \newunit%
}

\DeclareBibliographyDriver{article}{%
    \usebibmacro{bibindex}%
    \usebibmacro{begentry}%
    \usebibmacro{author/translator+others}%
    \setunit{\labelnamepunct}\newblock
    \usebibmacro{title}%
    \newunit
    \printlist{language}%
    \newunit\newblock
    \usebibmacro{byauthor}%
    \newunit\newblock
    \usebibmacro{bytranslator+others}%
    \newunit\newblock
    \printfield{version}%
    \newunit\newblock
    \printtext[doi/url-link]{%
        \usebibmacro{journal+issuetitle}%
        \newunit
        \usebibmacro{byeditor+others}%
        \newunit
        \usebibmacro{note+pages}%
        \newunit\newblock
        \iftoggle{bbx:isbn}
        {\printfield{issn}}
        {}%
    }%
    \newunit\newblock
    \printfield{pubstate}%
    \setunit{\addspace}%
    \printfield{year}%
    \newunit\newblock
    \usebibmacro{doi+eprint+url}%
    \newunit\newblock
    \printfield{addendum}%
    \setunit{\bibpagerefpunct}\newblock
    \usebibmacro{pageref}%
    \newunit\newblock
    \iftoggle{bbx:related}
        {\usebibmacro{related:init}%
        \usebibmacro{related}}
        {}%
    \usebibmacro{finentry}%
}

\makeatletter
\def\blindfootnote{\gdef\@thefnmark{}\@footnotetext}
\makeatother
\newcommand{\emaillink}[1]{\href{mailto:#1}{#1}}

\title{Lieb-Robinson bounds for irreversible dynamics via adaptive\\ Dyson expansion}
\date{}

\author{
    Carla Rubiliani%
    \texorpdfstring{%
        \,\orcidlink{0009-0000-6169-8658}
        \footnote{
            \foreignlanguage{ngerman}{
                Fachbereich Mathematik,
                Universität~Tübingen,\newline
                Auf~der~Morgenstelle~10,
                72076~Tübingen%
            },
            Germany%
        }
    }{}%
    \and
    Marius Wesle%
    \texorpdfstring{%
        \,\orcidlink{0009-0006-3365-4037}
        \footnotemark[1]
    }{}%
    \and
    Tom Wessel%
    \texorpdfstring{%
        \,\orcidlink{0000-0001-7593-0913}
        \footnote{
            Department of Mathematics, Technical University of Munich,\newline
            Boltzmannstr.~3,
            85748 Garching,
            Germany
        }
        \footnote{
            Munich Center for Quantum Science and Technology,
            80799 München,
            Germany
        }
    }{}%
}

\begin{document}

\bgroup
\hypersetup{hidelinks}
\maketitle\thispagestyle{empty}
\blindfootnote{
    \parbox[t]{.8\textwidth}{
        Email:
        \parbox[t]{.5\textwidth}{
            \emaillink{carla.rubiliani@uni-tuebingen.de}
            \\\emaillink{marius.wesle@uni-tuebingen.de}
            \\\emaillink{tom.wessel@tum.de}
        }
    }
}
\egroup

\begin{abstract}
    In this note, we study a broad class of time-dependent and irreversible dynamics for quantum spin systems in both finite and infinite volume.
    We develop a simple and flexible strategy for expanding commutators of the form \(\commutator{A,\tau_{s,t} \, B}\) in terms of self-avoiding paths of sets, where \(A\) and \(B\) are local operators and \(\tau_{s,t}\) is the Heisenberg time evolution.
    This \emph{adaptive Dyson expansion} requires no backwards time evolution and is thus also applicable to irreversible dynamics.
    Furthermore, we discuss how such an expansion directly implies Lieb-Robinson bounds.
    A central feature of our result is that the obtained Lieb-Robinson bounds are independent of any Hamiltonian and dissipative on-site terms of the generator.
\end{abstract}

\section{Introduction}

Understanding interacting many-body quantum systems is notoriously challenging due to the exponential scaling of the Hilbert space dimension with the number of degrees of freedom.
Lieb-Robinson bounds are a key tool in their rigorous analysis, as they provide locality estimates that control the propagation of information and correlations in space, reflecting the local structure of the underlying dynamics.
Originally proven by \textcite{LR1972} for finite-range spin systems, Lieb-Robinson bounds establish an effective speed bound for the propagation of information through the lattice.
More concretely, in their standard form for spin systems they provide bounds on the commutator
\begin{equation}
    \label{eq:LRB-short-range-intro}
    \norm[\big]{
        \commutator[\big]{
            A, \tau_{s,t} \, B
        }
    }
    \leq
    2 \, \norm{A} \, \norm{B} \, \min\List{\abs{X},\abs{Y}} \, \e^{b \paren*{v \, (t-s) - \dist{X,Y}}}
    .
\end{equation}
Here, \(\tau_{s,t}\) is the Heisenberg time evolution from time~\(s\) to time~\(t\), \(v\) is a constant depending on the operator norm of the local interaction terms and is commonly known as Lieb-Robinson velocity, and \(A\),~\(B\) are bounded operators supported on the subsets~\(X\),~\(Y\), respectively, of the lattice.
Lieb-Robinson bounds can be used to prove the existence of the thermodynamic limit of the time evolution, see e.g.~\cite{NOS2006,NSY2019,TWW2025}.
They attracted significant attention when they were identified as a powerful tool for resolving long-standing open problems in mathematical physics~\cite{Hastings2004lieb,HK2006,NS2006,Hastings2007area}, and have since become a standard tool in the study of many-body quantum lattice systems.
In the last 20~years, Lieb-Robinson bounds have been improved for special classes of systems, such as those with commuting~\cite{LW2025} or almost commuting~\cite{HHKL2021} interactions, one-dimensional systems~\cite{CL2019,CL2021}, and specific models such as the transverse field Ising model~\cite{WH2020}.
They have been generalized to short-range~\cite{HK2006,NS2006,NSY2019} and long-range~\cite{FGC2015,MKN2016,EMN2020,KS2020,TCE2020,TGB2021} quantum spin systems, to fermionic~\cite{HK2006,NS2006,NSY2017,TW2025} and bosonic~\cite{yin2022finite,kuwahara2021lieb,FLS2022,LRSZ2023,LRZ2025,kuwahara2024effective,KLR2025,deuchert2025dynamics,LR2026} lattice systems, and to spatially growing interactions~\cite{TWW2025}.
Furthermore, Lieb-Robinson bounds have been generalized to irreversible dynamics of open quantum systems with finite-range~\cite{Poulin2010} and more general interactions~\cite{NVZ2011,SEK2019}.

In this note, we present a simple and flexible method for establishing Lieb-Robinson bounds.
The method proceeds in two steps, which we outline now before giving a more precise description.
First, one derives an expansion of the commutator \(\commutator{A,\tau_{s,t} \, B}\) by iteratively applying Duhamel's formula.
In each step of the expansion, one compares the time evolution with a second evolution that depends on the previous steps.
For this reason, we call this procedure \emph{adaptive Dyson expansion}.
The result is a sum over self-avoiding paths of sets.
The second step is the derivation of Lieb-Robinson bounds, which can straightforwardly be proven from the adaptive Dyson expansion of the first step.

The method we develop is directly applicable to a broad class of dynamics generated by sums of local time-dependent dissipative terms.
In particular, this setting contains Lindbladian dynamics as a special case while also allowing for more general dissipative generators.
For example, when the on-site Hilbert spaces are infinite-dimensional, it allows for non-normal local generators, i.e. generators that are not ultraweakly continuous (see e.g.~\cite[Chapter 3]{Fagnola1999}).

The expansion is a generalization of the one derived for time-independent Hamiltonian systems by \textcite{CL2021}.
They obtain the expression by writing out the full Dyson series of the time evolution and performing a non-trivial resummation to handle non-self-avoiding contributions.
A simpler proof based on an iterative Duhamel scheme is given in~\cite[Appendix~A.3]{MLL2025}, where the self-avoiding structure follows directly from repeated applications of Duhamel's formula.
The adaptive Dyson expansion of this work was developed independently of~\cite{MLL2025} and follows a similar iterative proof.
We present a detailed and self-contained derivation of the expansion, and also establish the validity of the expansion in the thermodynamic limit.

Furthermore, our result yields Lieb-Robinson bounds that are independent of both Hamiltonian and dissipative on-site terms.
For Hamiltonian on-site terms, this independence is well known and follows from an interaction-picture argument using techniques developed in~\cite{NRSS2009}.
As pointed out in~\cite{NVZ2011}, on-site terms can be removed with the same techniques also for irreversible dynamics.
This, however, only applies to Hamiltonian on-site terms, since the interaction picture requires a backwards time evolution.
The present result shows, for the first time to the authors' knowledge, that dissipative on-site terms do not affect Lieb-Robinson bounds.
This property is particularly useful for the analysis of truncated boson models and used in~\cite{LM2026}.

This note also lays the foundation for a more involved version of the adaptive Dyson expansion, which we will use in an upcoming work~\cite{RWW2026} to prove Lieb-Robinson bounds with linear light cones for a broad class of long-range systems including fermions and irreversible dynamics.

Let us now give a more precise description of our result.
Although our results apply to a much broader class of dynamics generated by local, time-dependent, dissipative terms (see \cref{def:interaction} for the precise assumptions), we restrict our attention to a simpler setting in the introduction for ease of presentation.
Specifically, we consider a time evolution that is generated by a sum of local, time-independent terms~\(\Phi(Z)\) which are assumed to be Lindbladian terms in GKSL form acting on observables~\(O\) as
\begin{equation*}
    \Phi(Z) \, O
    =
    \sum_{j\in J_Z} \I \, \commutator{H_j(Z), O} + L_j(Z)^* \, O \, L_j(Z) - \tfrac{1}{2} \, \acomm[\big]{L_j(Z)^* \, L_j(Z), O}
    ,
\end{equation*}
where \(J_Z\) is a finite index set, \(\acomm{\cdot,\cdot}\) denotes the anticommutator, and the Hamiltonian terms \(H_j(Z)\) and jump operators \(L_j(Z)\) are bounded operators supported in \(Z\), with \(H_j(Z)\) self-adjoint.
For all sets \(X\) and \(Y\) in the lattice and for all observables~\(A\) and~\(B\) supported on \(X\) and~\(Y\), respectively, we prove the following exact expansion, where \(K\) acts on observables~\(O\) as \(K \, O = \commutator{A,O}\),
\begin{equation}
    \label{eq:exp-intro}
    \begin{aligned}
        K \, \tau_{s,t} \, B
        =
        \sum_{n=0}^\infty
        & \intlap_{s}^t \odif{u_1}
        \intlap_{u_1}^t \odif{u_2}
        \, \dotsi
        \intlap_{u_{n-1}}^t \odif{u_n}
        \, \sumstack{\gamma\in \SP_n(X,Y)}
        \\&
        K
        \, \tau_{s,u_1}
        \, \Phi(\gamma_1)
        \, \tau^{\gamma_{:0}}_{u_1,u_2}
        \, \Phi(\gamma_2)
        \, \dotsm
        \, \tau^{\gamma_{:n-2}}_{u_{n-1},u_n}
        \, \Phi(\gamma_{n})
        \, \tau^{\gamma_{:n-1}}_{u_n,t}
        \, B
        ,
    \end{aligned}
\end{equation}
in terms of self-avoiding paths \(\gamma = (X=\gamma_0,\gamma_1,\dotsc,\gamma_n,\gamma_{n+1}=Y)\) where \(\gamma_j\) are subsets of the lattice that intersect the preceding set \(\gamma_{j-1}\) but none of the earlier sets \(\gamma_{k}\) with \(k<j-1\).
The in-between time-evolutions \(\tau^{\gamma_{:j}}\) are generated by only the local terms that are disjoint from all previous sets in the path \(\gamma_{:j} = (\gamma_0,\dotsc,\gamma_j)\).
Note that~\eqref{eq:exp-intro} holds more generally for bounded operators \(K\) acting on observables, provided they act trivially outside~\(X\).
Importantly, the time-evolutions are all contractive, so that one immediately obtains the following Lieb-Robinson bound in terms of self-avoiding paths by the triangle inequality
\begin{equation}
    \label{eq:intro-expansion-norm-bound}
    \norm{
        K \, \tau_{s,t} \, B
    }
    \leq
    \norm{K}
    \, \norm{B}
    \, \sum_{n=0}^\infty
    \, \frac{(t-s)^n}{n!}
    \quad \sumstack[l]{\gamma\in \SP_n(X,Y)}
    \, \prod_{j=1}^n
    \, \norm{\Phi(\gamma_j)}
    .
\end{equation}
From this expression one can conclude a Lieb-Robinson bound as given in~\eqref{eq:LRB-short-range-intro} by standard tools, see \cref{thm:Lieb-Robinson-bound}.
We highlight that~\eqref{eq:intro-expansion-norm-bound} provides a stronger result through an example, which was already provided in~\cite{CL2021} for Hamiltonian systems:
Consider a model with only on-site terms and nearest-neighbour interactions on~\(\Z\), i.e.\ \(\Phi(Z) = 0\) unless \(Z=\List{z}\) or \(Z=\List{z,z+1}\) for some \(z\in \Z\), and with \({J \coloneq \sup_{z\in \Z} \, \norm{\Phi(\List{z,z+1})} < \infty}\) and \(\sup_{z\in \Z} \, \norm{\Phi(\List{z})}<\infty\).
Then, for any two finite, disjoint intervals%
\footnote{
    Here, an interval in \(\Z\) is the intersection of an interval in~\(\R\) with~\(\Z\).
}
\(X\),~\(Y\subseteq \Z\) there exists at most one self-avoiding path contributing to~\eqref{eq:intro-expansion-norm-bound}, which consists of the nearest-neighbour hoppings between \(X\) and \(Y\), and one obtains a Lieb-Robinson bound
\begin{equation*}
    \norm{
        K \, \tau_{s,t} \, B
    }
    \leq
    \norm{K}
    \, \norm{B}
    \, \frac{
        \paren[\big]{J \, (t-s)}^{\dist{X,Y}}
    }{
        \dist{X,Y}!
    }
    .
\end{equation*}

The note is organized as follows.
We start by introducing our set-up, including the definition of self-avoiding paths, in \cref{sec:math-setup}.
In \cref{sec:results}, we state our main results: the adaptive Dyson expansion, \cref{thm:adaptive-dyson-expansion}, the resulting Lieb-Robinson bounds, \cref{thm:Lieb-Robinson-bound,cor:Lieb-Robinson-bound-nearest-neighbour-spin-chain,thm:finite-range-Lieb-Robinson-bound}, and consequences for operator localization, \cref{cor:localization-of-operators}.
Moreover, we compare them to the literature and hint at the important steps of the proofs.
In \cref{sec:proof-adaptive-dyson-expansion}, we prove the adaptive Dyson expansion.
Using \cref{thm:adaptive-dyson-expansion}, we then prove the Lieb-Robinson bounds and related consequences in \cref{sec:proofs-Lieb-Robinson-bounds}.

\section{Mathematical setup}
\label{sec:math-setup}

In this section, we introduce the framework that we use to describe quantum lattice spin systems.
Along the way, we comment on how the more standard examples fit into this framework.
First, we introduce the lattices and algebraic description of quantum spin systems.
Afterwards, we introduce interactions as a generalization of the usual notion of interactions used in Hamiltonian systems, which allows us to also cover open quantum systems and discuss the time evolutions they generate.
We conclude the section by defining self-avoiding paths, which play an important role in the adaptive Dyson expansion.

\subsection{Algebraic description of lattice spin systems}

Let \((\Gamma,d)\) be a countable metric space.
We denote the set of finite subsets of~\(\Gamma\) by~\(\powerset(\Gamma)\).

\begin{definition}[Decay function]
    \label{def:decay-function}
    A non-increasing function \(F\colon \intervalco{0,\infty} \to \intervaloo{0,\infty}\) is called a \emph{decay function} on \(\Gamma\) if
    \begin{align}
        \norm{F}
        &\coloneq
        \sup_{x\in \Gamma} \sum_{y \in \Gamma} F\pdist{x,y}
        <
        \infty
        \label{eq:F-norm}
    \shortintertext{and}
        C_F
        &\coloneq
        \sup_{x,y\in \Gamma} \sum_{z\in\Gamma} \frac{F\pdist{x,z} \, F\pdist{z,y}}{F\pdist{x,y}}
        <
        \infty
        \label{eq:convolution-constant}
        .
    \end{align}
\end{definition}

Note that not every countable metric space admits a decay function.
Indeed, since \(F\) is positive and non-increasing, condition~\eqref{eq:F-norm} forces every ball of finite radius to contain finitely many sites.

\begin{remark}
    \label{rem:square-lattice-decay-function}
    On \(\Gamma=\Z^D\), or any subset thereof equipped with the \(\ell^1\) metric, the functions \(F_\alpha(r) \coloneq (1+r)^{-\alpha}\) and \(F_\alpha(r) \, \e^{-b r}\) are decay functions for \(\alpha>D\) and \(b>0\), see e.g.~\cite[Appendix]{NSY2019}.
\end{remark}

For each \(x \in \Gamma\), let \(\HS_x\neq \List{0}\) be some (possibly infinite-dimensional, non-separable) Hilbert space and for \(\Lambda \in \powerset(\Gamma)\) let
\begin{equation*}
    \tilde{\alg}_\Lambda
    \coloneq
    \calL \paren[\Big]{\bigotimes_{x\in \Lambda} \HS_x}
\end{equation*}
be the algebra of bounded operators on the tensor product Hilbert space.
For all \(\Lambda\),~\(\Lambda' \in \powerset(\Gamma)\) with \(\Lambda \subseteq \Lambda'\) the algebra \(\tilde{\alg}_\Lambda\) can naturally be embedded into \(\tilde{\alg}_{\Lambda'}\) by tensoring with the identity.
This allows us to take the inductive limit in the category of C*-algebras, which results in the quasi-local algebra \(\alg\).
In other words, up to a number of identifications, \(\alg\) is the norm-completion of the union of all \(\tilde{\alg}_\Lambda\).
For finite \(\Lambda \subseteq \Gamma\) we denote by \(\alg_\Lambda\) the embedding of \(\tilde{\alg}_\Lambda\) into \(\alg\) and for infinite \(\Lambda \subseteq \Gamma\) we define
\begin{equation*}
    \alg_{\Lambda}
    \coloneq
    \overline{\bigcup_{Z \in \powerset(\Lambda)} \alg_Z}
    .
\end{equation*}
Note that \(\alg_\Lambda \subseteq \alg_{\Lambda'}\) for all \(\Lambda \subseteq \Lambda' \subseteq \Gamma\), \(\alg_\emptyset= \C\cdot 1\) and \(\alg_\Gamma = \alg\).
We say that an operator \(A\in \alg\) is supported in \(\Lambda\subseteq \Gamma\) if \(A \in \alg_\Lambda\).

Denote by \(\calL(\alg)\) the algebra of bounded linear maps on \(\alg\).
We say that an operator \(K\in \calL(\alg)\) is localized in the set \(Z\subseteq \Gamma\) if
\begin{align*}
    K \, \alg_{\Lambda}
    &\subseteq
    \alg_{\Lambda\cup Z}
    \qquad\text{for all \(\Lambda\subseteq \Gamma\)}
\shortintertext{and}
    K \, \alg_{\Lambda}
    &=
    \List{0}
    \qquad\text{for all \(\Lambda\subseteq \Gamma\) with \(Z\cap \Lambda = \emptyset\)}
    .
\end{align*}
We denote by \(\calB_Z \subseteq \calL(\alg)\) the set of all operators localized in \(Z\).
For \(Z\subseteq Z' \subseteq \Gamma\) it holds that \(
    \List{0}
    =
    \calB_\emptyset
    \subseteq
    \calB_Z
    \subseteq
    \calB_{Z'}
    \subseteq
    \calB_\Gamma
    =
    \Set{K\in\calL(\alg) \given K \, 1=0}
\).

\subsection{Interactions and generated time evolutions}

Let \(I\subseteq \R\) be a fixed time interval.
Our results are formulated for so-called \emph{interactions}~\(\Phi\) that describe the local terms.
They are usually thought of as the (not necessarily convergent) sum of local terms \(\sum_{Z\in \powerset(\Gamma)} \Phi(Z,t)\), where each term \(\Phi(Z,t)\) acts on a finite support~\(Z\subseteq \Gamma\).

\begin{definition}[Interaction]
    \label{def:interaction}
    An \emph{interaction} is a map
    \begin{equation*}
        \Phi\colon \powerset(\Gamma)\times I \to \calL(\alg), \ (Z,t) \mapsto \Phi(Z,t)
        ,
    \end{equation*}
    such that the following conditions are satisfied:
    \begin{thmlist}
        \item
            \label{item:def-interaction:localized}
            \(\Phi(Z,t)\in \calB_Z\).
        \item
            \label{item:def-interaction:dissipative}
            \(\Phi(Z,t)\) is dissipative%
            \footnote{%
                An operator \(K\in \calL(\alg)\) is called dissipative if for all \(\lambda>0\) and all \(A\in \alg\) it holds that\\ \(\norm{(\lambda \, \id - K) \, A} \geq \lambda \, \norm{A}\).
                This is equivalent to \(\intervalco{0,\infty} \to \calL(\alg), \, t \mapsto \e^{tK}\) being a contraction semigroup by the Lumer-Phillips Theorem.
            }%
            .
        \item
            \label{item:def-interaction:continuous}
            For all \(Z\in \powerset(\Gamma)\) the map \(I \to \calL(\alg),\ t \mapsto \Phi(Z,t)\) is norm-continuous and uniformly bounded.
    \end{thmlist}

    \noindent
    We call \(\Phi\) \emph{finitely supported} if \(\sup_{t\in I} \, \norm{\Phi(Z,t)} = 0\) for all but finitely many \(Z\in \powerset(\Gamma)\).
    We say it has \emph{finite range}~\(R\geq 0\) if \(\sup_{t\in I} \, \norm{\Phi(Z,t)} = 0\) unless \(\diam{Z} \leq R\), and it is a \emph{\(k\)-body} interaction for some \(k\in\N_{>0}\) if \(\sup_{t\in I} \, \norm{\Phi(Z,t)} = 0\) unless \(\abs{Z} \leq k\).

    \noindent
    To control the total interaction strength at each site, we set
    \begin{equation}
        \label{eq:def-interaction-1-norm}
        \tnorm{\Phi}_1
        \coloneq
        \sup_{u\in I} \, \sup_{x\in \Gamma} \, \sumstack[r]{Z\in \powerset(\Gamma)\suchthat\\ x\in Z} \, \norm{\Phi(Z,u)}
        .
    \end{equation}

    \noindent
    For a decay function~\(F\) on \(\Gamma\), we define the \emph{interaction norm}
    \begin{equation}
        \label{eq:interaction-norm}
        \norm{\Phi(t)}_F
        \coloneq
        \sup_{x,y \in \Gamma}
        \sumstack{Z\in \powerset(\Gamma)\\ x,y\in Z}
        \frac{\norm{\Phi(Z,t)}}{F\pdist{x,y}}
        ,
    \end{equation}
    and say that~\(\Phi\) has \emph{\(F\)-decay} if
    \begin{equation*}
        \tnorm{\Phi}_F
        \coloneq
        \sup_{t\in I} \, \norm{\Phi(t)}_F
        <
        \infty
        .
    \end{equation*}
\end{definition}

In comparison to the interactions used to describe Hamiltonian systems whose terms are elements of~\(\alg_Z\), the interaction terms \(\Phi(Z,t)\) we consider here are bounded operators on~\(\alg\).
This is in order to include dissipative dynamics and simplify the notation in the expansion.
The extra conditions in \cref{def:interaction} are to ensure that these more general operators act locally in a physically meaningful way.
Importantly, usual Hamiltonian and Lindbladian dynamics satisfy these assumptions.

\begin{remark}
    In the usual description of open quantum systems with Markovian dynamics, the interaction \(\Phi\) corresponds to a time-dependent dual Lindblad generator in the Heisenberg picture.
    The most relevant class of interactions satisfying our assumptions are those whose local terms \(\Phi(Z,t)\) are given as local Lindbladians in GKSL form
    \begin{equation*}
        \Phi(Z,t) \, A
        =
        \sum_{j\in J_Z} \I \, \commutator{H_j(Z,t), A} + L_j(Z,t)^* \, A \, L_j(Z,t) - \tfrac{1}{2} \, \acomm[\big]{L_j(Z,t)^* \, L_j(Z,t), A}
        ,
    \end{equation*}
    where \(J_Z\) is a finite index set.
    If the Hamiltonian terms \(H_j(Z,t) \in \alg\) and jump operators \(L_j(Z,t)\in \alg\) are bounded operators supported in \(Z\), property~\itemref{item:def-interaction:localized} is automatically fulfilled.
    Furthermore, operators in the above GKSL form with \(H_j(Z,t)\) self-adjoint are always completely dissipative, which implies property~\itemref{item:def-interaction:dissipative} (cf.~\cite[Theorem~1]{NVZ2011}).
    Property~\itemref{item:def-interaction:continuous} follows if for each \(j\in J_Z\) both \(H_j(Z,t)\) and \(L_j(Z,t)\) are norm-continuous and uniformly bounded in \(t\).
\end{remark}

For any finitely supported interaction, we can now define the time evolution by a standard Dyson expansion, for which we introduce the generalized \(n\)-simplex
\begin{equation}
    \label{eq:definition-simplex}
    \Delta^n_J
    \coloneq
    \Set[\big]{
        (u_1,\dotsc,u_n)\in J^n
        \given
        u_1 \leq u_2 \leq \dotsb \leq u_n
    }
\end{equation}
for all intervals \(J\subseteq \R\) and the shorthand \(\Delta^n_{s,t} \coloneq \Delta^n_{\intervalcc{s,t}}\) that will be used later on.
Recall that the volume of an \(n\)-simplex for a finite interval \(J\subseteq \R\) is given by \(\abs{\Delta^n_J} = \abs{J}^n / n!\).

\begin{definition}[Time evolution generated by finitely supported interactions]
    \label{def:time-evolution-finitely-supported-case}
    Let \(\Phi\) be a finitely supported interaction.
    The \emph{time evolution generated by \(\Phi\)} is the map
    \begin{equation*}
        \tau\colon \Delta_I^2 \to \calL(\alg), \, (s,t) \mapsto \tau_{s,t}
    \end{equation*}
    defined by the time-ordered exponential
    \begin{equation}
        \label{eq:def-time-evolution-finitely-supported-case}
        \tau_{s,t}
        \coloneq
        \sum_{n=0}^\infty
        \intlap_{\Delta^n_{s,t}} \odif{u}
        \, \sumstack[r]{Z_1\in \powerset(\Gamma)}
        \, \Phi(Z_1,u_1)
        \, \dotsm
        \, \sumstack[lr]{Z_n\in \powerset(\Gamma)}
        \, \Phi(Z_n,u_n)
        ,
    \end{equation}
    where the \(n=0\) term is defined to be \(\id_{\alg}\), the identity map on \(\alg\).
\end{definition}

Since the interaction in \cref{def:time-evolution-finitely-supported-case} is finitely supported, \(t \mapsto \sum_{Z\in \powerset(\Gamma)} \, \Phi(Z,t)\) is norm-continuous and uniformly bounded and hence~\eqref{eq:def-time-evolution-finitely-supported-case} is absolutely convergent.
Moreover, the time evolution satisfies the following standard properties, which we prove for the sake of completeness in \cref{app:proof-time-evolution-properties}.

\begin{restatable}{proposition}{finiteevolutions}
    \label{prop:finite-evolutions}
    Let \(\Phi\) be a finitely supported interaction and let \(\tau\) be its generated evolution.
    It holds that
    \begin{thmlist}
        \item
            \label{item:prop-finite-evolutions:cocycle}
            \(\tau_{s,t} \, \tau_{t,u} = \tau_{s,u}\) for all \((s,t,u) \in \Delta^3_I\),
        \item
            \label{item:prop-finite-evolutions:identity}
            \(\tau_{s,s} = \id\) for all \(s \in I\),
        \item
            \label{item:prop-finite-evolutions:continuity}
            \(\tau\) is norm-continuous in both arguments,
        \item
            \label{item:prop-finite-evolutions:derivative-second}
            \(\tau\) is norm-differentiable in the second argument and\newline \(
                \partial_{t} \, \tau_{s,t}
                =
                \sum_{Z\in \powerset(\Gamma)} \tau_{s,t} \, \Phi(Z,t)\;
            \) for all \((s,t)\in \Delta^2_I\),
        \item
            \label{item:prop-finite-evolutions:derivative-first}
            \(\tau\) is norm-differentiable in the first argument and\newline \(
                \partial_{s} \, \tau_{s,t}
                =
                -\sum_{Z\in \powerset(\Gamma)} \Phi(Z,s) \, \tau_{s,t}\;
            \) for all \((s,t)\in \Delta^2_I\),
        \item
            \label{item:prop-finite-evolutions:contractive}
            \(\norm{\tau_{s,t}} \leq 1\) for all \((s,t)\in \Delta^2_I\).
    \end{thmlist}
\end{restatable}

Next, we define the time evolutions generated by general interactions as limits of restrictions to finite subsets.
Therefore, for any interaction~\(\Phi\) and set \(\Lambda\subseteq \Gamma\) we define the restricted interaction
\begin{equation}
    \label{eq:restricted-interaction}
    \Phi\rvert_\Lambda
    \colon
    (Z,t)
    \mapsto
    \begin{cases*}
        \Phi(Z,t)
          & if \(Z\subseteq \Lambda\),
        \\
        0 & otherwise.
    \end{cases*}
\end{equation}

\begin{definition}[Time evolution generated by general interactions]
    \label{def:time-evolution-general-case}
    Let \(\Phi\) be an interaction.
    For each \(\Lambda\in \powerset(\Gamma)\), consider the restricted interaction \(\Phi\rvert_\Lambda\) and let \(\tau^\Lambda\) be its generated evolution.
    Then, for all \(B\in \alg\) and \(s\),~\(t\in I\) with \(s \leq t\), the \emph{time evolution generated by \(\Phi\)} is defined pointwise as the limit of the net \((\tau_{s,t}^\Lambda \, B)_{\Lambda \in \powerset(\Gamma)}\) if it converges.
\end{definition}

Whenever the defining net converges for every \(B\in\alg\), the resulting map is linear and contractive, being a pointwise limit of the contractions \(\tau^\Lambda_{s,t}\) (\cref{item:prop-finite-evolutions:contractive}).
Moreover, for finitely supported \(\Phi\), the net is eventually constant, and hence the definition agrees with \cref{def:time-evolution-finitely-supported-case}.

\FloatBarrier

\subsection{Self-avoiding paths}

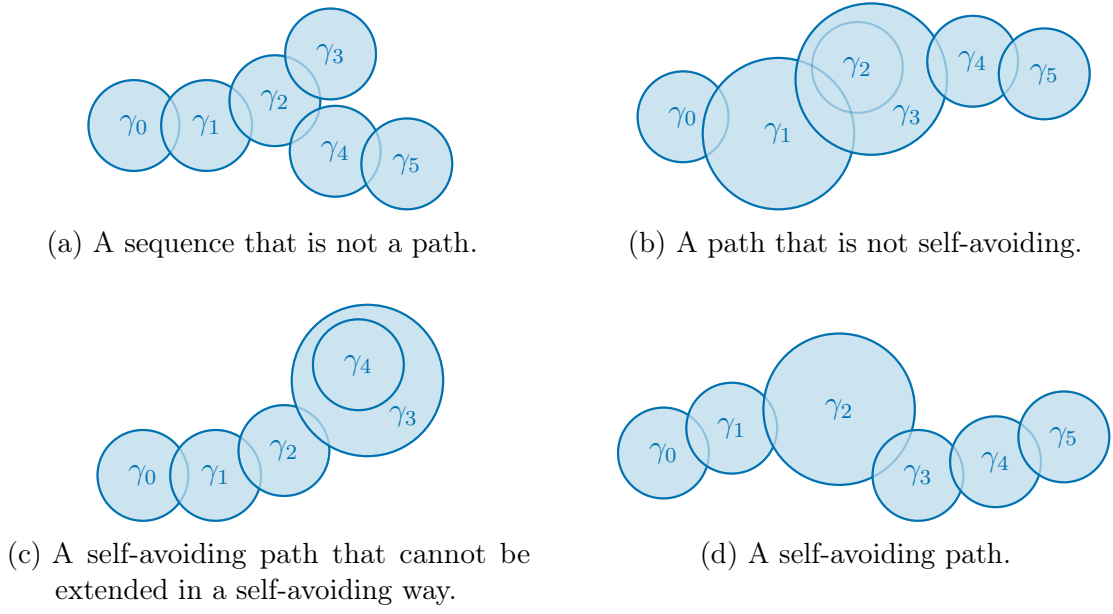
\begin{figure}
    \setlength{\ul}{.8cm}
    \centering
    \subcaptionbox{
        A sequence that is not a path.
        \label{subfig:sequence-not-path}
    }[.45\textwidth]{
        \begin{tikzpicture}
            \path[every node/.style={set=col_1}]
            (0,0) node {\(\gamma_0\)}
            ++(0:1.2) node {\(\gamma_1\)}
            ++(20:1.2) node {\(\gamma_2\)}
            +(40:1.2) node {\(\gamma_3\)}
            ++(-40:1.3) node {\(\gamma_4\)}
            ++(-10:1.2) node {\(\gamma_5\)}
            ;
        \end{tikzpicture}
    }
    \qquad
    \subcaptionbox{
        A path that is not self-avoiding.
        \label{subfig:non-self-avoiding-path}
    }[.45\textwidth]{
        \begin{tikzpicture}
            \path[every node/.style={set=col_1}]
            (0,0) node {\(\gamma_0\)}
            ++(-10:1.6) node[big] {\(\gamma_1\)}
            ++(40:1.7) node {\(\gamma_2\)}
            ++(-40:.3) node[big] (gamma3) {}
            ++(10:1.7) node {\(\gamma_4\)}
            ++(-10:1.2) node {\(\gamma_5\)}
            ;
            \path
            (gamma3.south east) node[above left, set label=col_1] {\(\gamma_3\)}
            ;
        \end{tikzpicture}
    }
    \\[\baselineskip]
    \subcaptionbox{
        A self-avoiding path that cannot be extended in a self-avoiding way.
        \label{subfig:self-avoiding-path-that-cannot-be-extended}
    }[.45\textwidth]{
        \begin{tikzpicture}
            \path[every node/.style={set=col_1}]
            (0,0) node {\(\gamma_0\)}
            ++(0:1.2) node {\(\gamma_1\)}
            ++(20:1.2) node {\(\gamma_2\)}
            ++(40:1.8) node[big] (gamma3) {}
            ++(120:.3) node {\(\gamma_4\)}
            ;
            \path
            (gamma3.south east) node[above left, set label=col_1] {\(\gamma_3\)}
            ;
        \end{tikzpicture}
    }
    \qquad
    \subcaptionbox{
        A self-avoiding path.
        \label{subfig:self-avoiding-path}
    }[.45\textwidth]{
        \begin{tikzpicture}[every node/.style={set=col_1}]
            \path
            (0,0) node {\(\gamma_0\)}
            ++(20:1.2) node {\(\gamma_1\)}
            ++(10:1.8) node[big] {\(\gamma_2\)}
            ++(-40:1.7) node {\(\gamma_3\)}
            ++(10:1.3) node {\(\gamma_4\)}
            ++(20:1.2) node {\(\gamma_5\)}
            ;
        \end{tikzpicture}
    }
    \caption{
        Examples of sequences to illustrate the definitions of paths and self-avoiding paths.
    }
    \label{fig:path-examples}
\end{figure}

To formulate the adaptive Dyson expansion we need to define self-avoiding paths.
We begin by defining the set of all sequences in \(\Gamma\) as
\begin{equation*}
    \Seq
    \coloneq
    \Set*{
        (\gamma_{j})_{j\in J}
        \given
        \begin{aligned}
            & J = \List{0,\dotsc,l} \ \text{for some}\ l \in \N_{0} \, \text{or} \, J=\emptyset,
            \\&\forall j \in J \colonpunct\ \gamma_j\in\powerset(\Gamma)
        \end{aligned}
    }
    ,
\end{equation*}
where we explicitly include the empty sequence \(( \, )\) as the family indexed by \(\emptyset\).
For each \(\gamma =(\gamma_j)_{j\in J} \in \Seq\) we denote by \(\ell_\gamma\) the maximal index of \(\gamma\) (i.e.\ \(\gamma=(\gamma_j)_{j\in \List{0,\dotsc,\ell_\gamma}}\)) with \(\ell_{( \, )} \coloneq -1\).
Sequences can also be written as tuples \(\gamma = (\gamma_0,\dotsc,\gamma_{\ell_\gamma})\) with a zero-indexing convention.
For every sequence \(\gamma\) and \(k \in \Z\) with \(-1 \leq k \leq \ell_\gamma\), we define
\begin{equation*}
    \gamma_{:k}
    \coloneq
    \begin{cases*}
        ( \, )                     & if \(k=-1\)
        \\
        (\gamma_0,\dotsc,\gamma_k) & if \(0\leq k \leq \ell_\gamma\).
    \end{cases*}
\end{equation*}

We use the following notation: for \(\gamma=(\gamma_j)_{j\in J} \in \Seq\) and \(Z\in \powerset(\Gamma)\) we define the sequence \(\gamma \oplus Z \coloneq (\gamma_0,\dotsc,\gamma_{\ell_\gamma},Z)\), and we write \(\gamma \perp Z\) if \(\gamma_j \cap Z = \emptyset\) for all \(j\in J\).

A sequence \(\gamma=(\gamma_j)_{j\in J}\) is said to be a \emph{path} if consecutive sets overlap, i.e.
\begin{equation*}
    \forall j \in J \colonpunct \quad j+1 \in J \implies \gamma_{j} \cap \gamma_{j+1} \neq \emptyset
    ,
\end{equation*}
and a path \(\gamma=(\gamma_j)_{j\in J}\) is called \emph{self-avoiding} if only consecutive sets overlap, i.e.
\begin{equation*}
    \forall j,k \in J \colonpunct \quad
    \gamma_j\cap\gamma_k\neq \emptyset \implies \abs{j-k} \in \List{0,1}
    .
\end{equation*}
We denote the set of all paths by \(\Path \subseteq \Seq\) and the set of self-avoiding paths by \(\SP\).
Some examples to illustrate these definitions are depicted in \cref{fig:path-examples}.
For \(n\in \N_0\) we define the set of (self-avoiding) paths that start at a given set~\(X \in \powerset(\Gamma)\) and contain \(n\) subsequent sets as
\begin{equation*}
    \Path_n(X)
    \coloneq
    \Set[\big]{
        \gamma \in \Path
        \given
        \ell_\gamma = n, \
        \gamma_0=X
    }
\end{equation*}
and \(\SP_n(X) \coloneq \Path_n(X) \cap \SP\).
Likewise, we define the set of (self-avoiding) paths that connect two given sets~\(X\),~\(Y \in \powerset(\Gamma)\) via \(n\) intermediate sets as
\begin{equation*}
    \Path_n(X,Y)
    \coloneq
    \Set[\big]{
        \gamma \in \Path
        \given
        \ell_\gamma = n+1, \
        \gamma_0=X,\
        \gamma_{n+1}=Y
    }
\end{equation*}
and \(\SP_n(X,Y) \coloneq \Path_n(X,Y) \cap \SP\).

\section{Results}
\label{sec:results}

Our main result is an adaptive Dyson expansion, which allows to exactly rewrite the time evolution in terms of self-avoiding paths.
The proof is given in \cref{sec:proof-adaptive-dyson-expansion}.

\begin{theorem}[Adaptive Dyson expansion]
    \label{thm:adaptive-dyson-expansion}
    Let \(\Phi\) be an interaction satisfying
    \begin{equation}
        \label{eq:weighted-interaction-summability}
        C_\Phi
        \coloneq
        \sup_{t\in I} \, \sup_{x\in\Gamma}
        \, \sumstack[r]{Z\in\powerset(\Gamma)\\ x\in Z}
        \, \abs{Z} \, \norm{\Phi(Z,t)}
        <
        \infty
        .
    \end{equation}
    Then \(\Phi\) generates a time evolution~\(\tau\) in the sense of \cref{def:time-evolution-general-case}.
    For every self-avoiding path \(\eta \in\SP\), the interaction
    \begin{equation}
        \label{eq:thm-path-dependent-interactions}
        \Phi^{\eta}(Z,u)
        \coloneq
        \begin{cases*}
            \Phi(Z,u) & if \(\eta \perp Z\)
            \\
            0         & otherwise
        \end{cases*}
    \end{equation}
    also generates a time evolution, denoted by~\(\tau^{\eta}\).

    For all \(s\),~\(t \in I\) with \(s \leq t\), \(X\),~\(Y\in \powerset(\Gamma)\), \(K \in \calB_X\) and \(B\in \alg_Y\), we have the following absolutely convergent expansion
    \begin{equation}
        \label{eq:thm-adaptive-dyson-expansion}
        K \, \tau_{s,t} \, B
        =
        K
        \, \sum_{n=0}^\infty
        \, \intlap_{\Delta^n_{s,t}} \odif{u}
        \, \sumstack[r]{\gamma\in \SP_n(X,Y)}
        \,
        \begin{aligned}[t]
            \tau_{s,u_1}
            \, \Phi(\gamma_1,u_1)
            & \, \tau^{\gamma_{:0}}_{u_1,u_2}
            \, \Phi(\gamma_2,u_2)
            \\&\dotsm
            \, \tau^{\gamma_{:n-2}}_{u_{n-1},u_n}
            \, \Phi(\gamma_{n},u_n)
            \, \tau^{\gamma_{:n-1}}_{u_n,t}
            \, B
            ,
        \end{aligned}
    \end{equation}
    where the \(n=0\) term is defined to be \(K \, \sum_{\gamma\in \SP_0(X,Y)} \, \tau_{s,t} \, B\).

    Furthermore, if \(X\) and \(Y\) are disjoint, then
    \begin{equation}
        \label{eq:thm-adaptive-dyson-expansion-norm-bound}
        \norm{K \, \tau_{s,t} \, B}
        \leq
        \norm{K} \, \norm{B}
        \, \sum_{n=1}^{\infty}
        \, \sumstack{\gamma\in\SP_n(X,Y)}
        \, \intlap_{\Delta^n_{s,t}} \odif{u}
        \, \prod_{j=1}^n
        \, \norm{\Phi(\gamma_j,u_j)}
        .
    \end{equation}
\end{theorem}

Notice that~\eqref{eq:thm-adaptive-dyson-expansion-norm-bound} follows directly from~\eqref{eq:thm-adaptive-dyson-expansion} by taking norms, applying the triangle inequality, and using contractivity of the time evolutions.
Moreover, condition~\eqref{eq:weighted-interaction-summability} is satisfied for interactions with \(F\)-decay by
\(
    C_\Phi
    \leq
    \norm{F} \, \tnorm{\Phi}_F
\)
and for \(k\)-body interactions by
\(
    C_\Phi
    \leq
    k \, \tnorm{\Phi}_1
\).

As mentioned in the introduction, a similar expansion was first proven in~\cite{CL2021} for time-independent Hamiltonian quantum systems in finite volume with an involved proof, in which the evolution is first completely expanded using the exponential series and then resummed to obtain~\eqref{eq:thm-adaptive-dyson-expansion}.
In comparison to our expansion, the in-between evolutions in~\cite{CL2021} have a slightly more complicated structure and are generated by interactions~\(\Phi^\gamma\) that have non-zero terms supported within the sets of the path.
Compared to~\cite[Appendix~A.3]{MLL2025}, we provide a more systematic presentation of the derivation of the self-avoiding path expansion and treat a much more general class of dynamics including non-GKSL dissipative local generators and the thermodynamic limit.
We furthermore examine which interaction terms can occur along the self-avoiding paths.
This analysis identifies the terms that enter the resulting Lieb-Robinson bounds and makes explicit why on-site terms and terms supported entirely within either endpoint set do not contribute.

Estimating the sum over self-avoiding paths in~\eqref{eq:thm-adaptive-dyson-expansion-norm-bound}, using either the \(F\)-decay of the interactions or their \(k\)-body and finite-range structure, we obtain Lieb-Robinson bounds.
While the standard proofs of Lieb-Robinson bounds, see e.g.~\cite{NSY2019,NVZ2011}, sum over all paths, restricting to self-avoiding paths can sharpen the bounds and allows the estimates to be adapted to the system under consideration.
We present several applications of this flexibility below, with proofs in \cref{sec:proofs-Lieb-Robinson-bounds}.

As observed in~\cite{CLY2023}, self-avoidance yields much stronger Lieb-Robinson bounds for nearest-neighbour interactions in one-dimensional systems.

\begin{corollary}
    \label{cor:Lieb-Robinson-bound-nearest-neighbour-spin-chain}
    Suppose \(\Gamma=\Z\) with the \(\ell^1\)-distance.
    Let \(\Phi\) be a nearest-neighbour interaction, i.e.\ \(\Phi(Z,u)=0\) unless \(\diam{Z} \leq 1\), and assume that \({\tnorm{\Phi}_1 < \infty}\).
    Then \(\Phi\) generates a time evolution~\(\tau\).

    For non-empty finite sets \(X\),~\(Y\subseteq\Z\), write \(X<Y\) if \(\max X<\min Y\).
    For all such \(X\),~\(Y\) with \(X<Y\) or \(Y<X\), all \(s\),~\(t\in I\) with \(s\leq t\), \(K\in\calB_X\), and \(B\in\alg_Y\), it holds that
    \begin{equation*}
        \norm{K \, \tau_{s,t} \, B}
        \leq
        \norm{K}
        \, \norm{B}
        \, \frac{
            \paren[\big]{
                J \, (t-s)
            }^{\dist{X,Y}}
        }{
            \dist{X,Y}!
        }
        ,
    \end{equation*}
    where \(
        J \coloneq \sup_{u\in \intervalcc{s,t}}
        \, \sup_{z\in \Lambda}
        \, \norm[\big]{\Phi(\List{z,z+1}, u)}
    \) and
    \begin{equation*}
        \Lambda
        =
        \text{\small\(\displaystyle
                \begin{cases*}
                    \List{\max X, \max X + 1, \dotsc, \min Y -1}
                     & if \(X<Y\),
                    \\
                    \List{\max Y, \max Y + 1, \dotsc, \min X -1}
                     & if \(Y<X\).
                \end{cases*}
            \)}
    \end{equation*}
\end{corollary}

The proof uses the fact that there is at most one self-avoiding path from \(X\) to \(Y\) with non-zero contribution in~\eqref{eq:thm-adaptive-dyson-expansion-norm-bound}, and its intermediate sets consist of two neighbouring sites.
The complete argument is provided in \cref{sec:proof-Lieb-Robinson-bound-nearest-neighbour-spin-chain}.

We next introduce some notation and state Lieb-Robinson bounds in more general settings.
For a decay function~\(F\) on \(\Gamma\), an interaction~\(\Phi\) with \(F\)-decay, and \(s\),~\(t\in I\) with \(s\leq t\), set
\begin{equation*}
    I_\Phi(s,t)
    \coloneq
    \intlap_s^t \odif{u} \, \norm{\Phi(u)}_F
    \leq
    \tnorm{\Phi}_F \, (t-s)
    .
\end{equation*}
For finite sets \(X\),~\(Y\subseteq\Gamma\), let \(\tilde{\Phi}\) be the interaction obtained by removing the on-site terms and the terms supported entirely in \(X\) or \(Y\).
More precisely,
\begin{equation}
    \label{eq:thm-interaction-without-on-site-terms}
    \tilde{\Phi}(Z,t)
    \coloneq
    \begin{cases*}
        0
         & if \(Z\subseteq X\), \(Z\subseteq Y\), or \(\abs{Z}=1\),
        \\
        \Phi(Z,t)
         & otherwise.
    \end{cases*}
\end{equation}

With this notation, we can now state the following result, which is a direct consequence of the adaptive Dyson expansion.

\begin{theorem}[Lieb-Robinson bound]
    \label{thm:Lieb-Robinson-bound}
    Let \(F\) be a decay function on \(\Gamma\), and let \(\Phi\) be an interaction with \(F\)-decay.
    Let \(\tau\) be its generated evolution.
    For all \(s\),~\(t\in I\) with \(s\leq t\), disjoint \(X\),~\(Y\in\powerset(\Gamma)\), \(K\in\calB_X\), and \(B\in\alg_Y\), it holds that
    \begin{equation}
        \label{eq:thm-Lieb-Robinson-bound}
        \norm{K \, \tau_{s,t} \, B}
        \leq
        \norm{K} \, \norm{B} \, C_F^{-1}
        \, \paren[\big]{\e^{C_F \, I_{\tilde{\Phi}}(s,t)}-1}
        \, \sum_{x\in X} \, \sum_{y\in Y} \, F\pdist{x,y}
        .
    \end{equation}
\end{theorem}

The appearance of \(\tilde{\Phi}\) rather than \(\Phi\) in~\eqref{eq:thm-Lieb-Robinson-bound} is a consequence of the self-avoidance condition in~\eqref{eq:thm-adaptive-dyson-expansion-norm-bound}.
Indeed, neither on-site terms nor terms supported entirely in \(X\) or \(Y\) can occur among the intermediate sets of a self-avoiding path connecting disjoint \(X\) and \(Y\).
The bound is therefore independent of all such terms.
To our knowledge, this is the first proof of Lieb-Robinson bounds for irreversible dynamics that are independent of the on-site dissipative terms.

\begin{remark}
    \label{rem:square-lattice-Lieb-Robinson-bound}
    Consider again \(\Gamma = \Z^D\) equipped with the \(\ell^1\)-distance, and fix \(\alpha > D\).
    We use the decay functions introduced in \cref{rem:square-lattice-decay-function}.
    There exists a constant \(C\) such that
    \begin{equation}
        \label{eq:rem-tail-sum-on-square-lattice}
        \sumstack[r]{
            y\in \Gamma \suchthat\\
            \dist{x,y} \geq r
        } \, F_\alpha\pdist{x,y}
        \leq
        C \, F_{\alpha-D}(r)
        ,
    \end{equation}
    for all \(x\in \Gamma\) and \(r \geq 0\).
    Hence, the double sum in~\eqref{eq:thm-Lieb-Robinson-bound} is bounded as
    \begin{equation*}
        \sum_{x\in X} \, \sum_{y\in Y} \, F_\alpha\pdist{x,y}
        \leq
        C \, \min\List{\abs{X},\abs{Y}} \, F_{\alpha-D} \pdist{X,Y}
    \end{equation*}
    by applying~\eqref{eq:rem-tail-sum-on-square-lattice} to either of the sums.

    For the exponential decay \(F(r) = F_\alpha(r) \, \e^{-b r}\) the same argument produces an additional factor \(\e^{-b \, \dist{X,Y}}\), so that~\eqref{eq:thm-Lieb-Robinson-bound} gives
    \begin{equation*}
        \norm{K \, \tau_{s,t} \, B}
        \leq
        \norm{K} \, \norm{B} \, C \, C_F^{-1} \, \min\List{\abs{X},\abs{Y}} \, \e^{b \, \paren*{\xmathstrut{.1}v \, (t-s) - \dist{X,Y}}}
        ,
    \end{equation*}
    which resembles the form~\eqref{eq:LRB-short-range-intro} with velocity \(v = C_F \, \tnorm{\tilde{\Phi}}_F / b\).
\end{remark}

We next consider finite-range \(k\)-body interactions with \(\tnorm{\Phi}_1<\infty\).
As in the \(F\)-decay case, \cref{thm:adaptive-dyson-expansion} gives the existence of the dynamics and the adaptive Dyson expansion, from which we derive a Lieb-Robinson bound.

\begin{theorem}[Lieb-Robinson bounds for finite-range interactions]
    \label{thm:finite-range-Lieb-Robinson-bound}
    Let \(k\in\N_{>0}\), \(R>0\) and \(\Phi\) be a \(k\)-body interaction of finite range~\(R\), and assume that \(\tnorm{\Phi}_1<\infty\).
    Let \(\tau\) be its generated evolution.

    For all \(s\),~\(t\in I\) with \(s\leq t\), disjoint non-empty finite sets \(X\),~\(Y\subseteq\Gamma\), \(K\in\calB_X\), and \(B\in\alg_Y\), it holds that
    \begin{equation}
        \label{eq:thm-finite-range-Lieb-Robinson-bound}
        \norm{K \, \tau_{s,t} \, B}
        \leq
        \norm{K} \, \norm{B}
        \, \min\List{\abs{X},\abs{Y}}
        \, k^{-1}
        \, \paren[\big]{\e^{v \, (t-s)} - 1}
        \, \e^{- \dist{X,Y}/R}
        ,
    \end{equation}
    where \(v \coloneq \e \, k \, \tnorm{\tilde{\Phi}}_1\) and \(\tilde{\Phi}\) is defined in~\eqref{eq:thm-interaction-without-on-site-terms}.
\end{theorem}

Where the \(F\)-decay bound uses the convolution property~\eqref{eq:convolution-constant}, we now use the body bound~\(k\) and the condition \(\tnorm{\Phi}_1<\infty\) to control the path sums.
The spatial decay supplied by~\(F\) is replaced by exponential decay in \(\dist{X,Y}/R\), obtained from the finite range~\(R\).

\begin{remark}
    For \(\Gamma=\Z^D\) with the \(\ell^p\) distance, \(1\leq p\leq\infty\), every set of diameter at most \(R\in \N_{>0}\) fits into a cube of side length \(R\), containing at most \((R+1)^D\) lattice points.
    Thus, for an interaction of range \(R\), one may always take \(k=(R+1)^D\).
    Moreover,
    \begin{equation*}
        \tnorm{\Phi}_1
        \leq
        \sum_{j=1}^k
        \, \binom{(2R+1)^D}{j}
        \, \sup_{u\in I}
        \, \sup_{Z\in\powerset(\Gamma)}
        \, \norm{\Phi(Z,u)}
        ,
    \end{equation*}
    since every contributing set \(Z\ni x\) is contained in the ball of radius \(R\) around \(x\), which contains at most \((2R+1)^D\) lattice points.
\end{remark}

Another concrete application where self-avoiding paths are useful is the localization of operators.
Compared to the result obtained with usual Lieb-Robinson bounds, terms supported outside the region in which one localizes do not affect the bound.
This is physically relevant, because distant terms in the interaction should not influence the local spreading of observables.
For the proof, which is given in \cref{sec:proof-localization-of-operators}, one combines \cref{thm:Lieb-Robinson-bound} with standard arguments.
The region in which we localize is given by the \(r\)-neighbourhood defined as
\begin{equation}
    \label{eq:r-neighbourhood}
    Z[r]
    \coloneq
    \Set{x\in \Gamma \given \dist{Z,x} \leq r}
    \qquad\text{for all \(Z\subseteq \Gamma\)}
    .
\end{equation}

\begin{corollary}[Localization of operators]
    \label{cor:localization-of-operators}
    Let \(F\) be a decay function on \(\Gamma\), and let \(\Phi\) be an interaction with \(F\)-decay.
    Let \(\tau\) be its generated evolution.
    For all \(s\),~\(t \in I\) with \(s \leq t\), \(Y\in \powerset(\Gamma)\), \(B\in \alg_Y\), \(r\geq 0\), there exists an approximation \(B_r\in \alg_{Y[r]}\) such that
    \begin{equation*}
        \norm{
            \tau_{s,t} \, B - B_r
        }
        \leq
        4
        \, \norm{B}
        \, C_F^{-1}
        \, \paren[\big]{\e^{C_F \, I_{\tilde{\Phi}}(s,t)}-1}
        \, \sum_{x\in Y} \, \sumstack[r]{y\in \Gamma\setminus Y[r]} \, F\pdist{x,y}
        ,
    \end{equation*}
    with
    \begin{equation*}
        \tilde{\Phi}(Z,u)
        =
        \begin{cases*}
            0
             & if \(Z\subseteq Y\) or \(Z\subseteq \Gamma\setminus Y[r]\) or \(\abs{Z}=1\),
            \\
            \Phi(Z,u)
             & otherwise.
        \end{cases*}
    \end{equation*}
\end{corollary}

\section{Proof of the adaptive Dyson expansion}
\label{sec:proof-adaptive-dyson-expansion}

We prove \cref{thm:adaptive-dyson-expansion} in two steps.
First, in \cref{sec:proof-finite-volume}, we establish the expansion~\eqref{eq:thm-adaptive-dyson-expansion} for finitely supported interactions.
Then, in \cref{sec:proof-thermodynamic-limit}, we prove existence of the dynamics for interactions satisfying~\eqref{eq:weighted-interaction-summability} by using the expansion for the restricted interactions \(\Phi\rvert_\Lambda\).
Passing to the thermodynamic limit in the expansion completes the proof of~\eqref{eq:thm-adaptive-dyson-expansion}.

Before we continue, notice that \(X=\emptyset\) implies \(K = 0\) and \(\SP_n(X,Y) = \emptyset\).
Likewise, \(Y=\emptyset\) implies \(B = c \, 1\), therefore \(K \, \tau_{s,t} \, B = K \, B = 0\),  and \(\SP_n(X,Y) = \emptyset\).
It follows that both sides of~\eqref{eq:thm-adaptive-dyson-expansion} vanish in both cases.
Hence, we assume that \(X\) and \(Y\) are non-empty in the following.

\subsection{Finitely supported interactions}
\label{sec:proof-finite-volume}

Let \(\Phi\) be a finitely supported interaction and denote by \(\tau\) its generated evolution.
Fix \((s,t)\in \Delta^2_I\), non-empty \(X\),~\(Y\in \powerset(\Gamma)\), \(K \in \calB_X\) and \(B\in \alg_Y\).

We begin by establishing a truncated version of the claimed expansion with finitely many terms plus a remainder
\begin{equation}
    \label{eq:proof-finite-volume-truncated-expansion}
    K \, \tau_{s,t} \, B
    =
    \sum_{m=0}^n \calQ_m + \calR_n
    ,
\end{equation}
where
\begin{equation*}
    \calQ_0
    \coloneq
    K
    \, \sumstack[lr]{\gamma \in \SP_0(X,Y)}
    \, \tau^{\gamma_{:-1}}_{s,t} \, B
    \qquadtext{and}
    \calR_0
    \coloneq
    K \, \sumstack[lr]{
        \gamma \in \SP_0(X) \suchthat\\
        \gamma \perp Y
    }
    \, \tau^{\gamma_{:-1}}_{s,t} \, B
    ,
\end{equation*}
and for all \(n \in \N_{>0}\)
\begin{align*}
    \calQ_n
    &\coloneq
    K
    \, \intlap_{\Delta^n_{s,t}} \odif{u}
    \, \sumstack{\gamma \in \SP_n(X,Y)}
    \, \tau^{\gamma_{:-1}}_{s,u_1}
    \, \Phi(\gamma_1,u_1)
    \, \tau^{\gamma_{:0}}_{u_1,u_2}
    \, \Phi(\gamma_2,u_2)
    \, \dotsm
    \, \tau^{\gamma_{:n-2}}_{u_{n-1},u_n}
    \, \Phi(\gamma_{n},u_n)
    \, \tau^{\gamma_{:n-1}}_{u_n,t}
    \, B
    ,
    \\
    \calR_n
    &\coloneq
    K
    \, \intlap_{\Delta^n_{s,t}} \odif{u}
    \, \sumstack{
        \gamma \in \SP_n(X) \suchthat\\
        \gamma \perp Y
    }
    \, \tau^{\gamma_{:-1}}_{s,u_1}
    \, \Phi(\gamma_1,u_1)
    \, \tau^{\gamma_{:0}}_{u_1,u_2}
    \, \Phi(\gamma_2,u_2)
    \, \dotsm
    \, \tau^{\gamma_{:n-2}}_{u_{n-1},u_n}
    \, \Phi(\gamma_{n},u_n)
    \, \tau^{\gamma_{:n-1}}_{u_n,t}
    \, B
    .
\end{align*}
Recall that \(\gamma_{:-1} = ( \, )\) and \(\tau^{( \, )}_{s,t} = \tau_{s,t}\) by definition.
After establishing~\eqref{eq:proof-finite-volume-truncated-expansion}, we will show that the remainder vanishes and the series converges as \(n\to\infty\).

\minisec{Expansion with remainder}
We now prove~\eqref{eq:proof-finite-volume-truncated-expansion} by induction over \(n \in \N_0\).
The base case \(n=0\) follows by noting that
\begin{equation*}
    \SP_0(X,Y)
    =
    \begin{cases*}
        \List{(X,Y)} & \(X\cap Y \neq \emptyset\)
        \\
        \emptyset    & \(X\cap Y = \emptyset\)
    \end{cases*}
\end{equation*}
and
\begin{equation*}
    \Set{\gamma \in \SP_0(X) \given \gamma \perp Y}
    =
    \begin{cases*}
        \emptyset  & \(X\cap Y \neq \emptyset\)
        \\
        \List{(X)} & \(X\cap Y = \emptyset\).
    \end{cases*}
\end{equation*}
Therefore, we have
\begin{equation*}
    \calQ_0
    =
    \begin{cases*}
        K \, \tau_{s,t} \, B & \(X\cap Y \neq \emptyset\)
        \\
        0                    & \(X\cap Y = \emptyset\)
    \end{cases*}
    \qquadtext{and}
    \calR_0
    =
    \begin{cases*}
        0                    & \(X\cap Y \neq \emptyset\)
        \\
        K \, \tau_{s,t} \, B & \(X\cap Y = \emptyset\).
    \end{cases*}
\end{equation*}

Assume that the truncated expansion~\eqref{eq:proof-finite-volume-truncated-expansion} holds for some \(n\in\N_0\).
To prove it at order \(n+1\), it suffices to show that
\begin{equation*}
    \calR_n
    =
    \calQ_{n+1} + \calR_{n+1}
    .
\end{equation*}

We first consider the case \(n=0\).
If \(X\cap Y\neq\emptyset\), then \(\calR_0=\calQ_1=\calR_1=0\).
Otherwise, \(\calR_0=K \, \tau_{s,t} \, B\), and comparing \(\tau\) and \(\tau^{(X)}\) using Duhamel's formula yields
\begin{align*}
    \calR_0
    &=
    K \, \tau_{s,t} \, B
    \\&=
    K \, \tau^{(X)}_{s,t} \, B
    +
    \intlap_s^t \odif{u}
    \, K
    \, \tau_{s,u}
    \, \sumstack[lr]{
        Z\in\powerset(\Gamma)\\
        Z\cap X\neq\emptyset
    }
    \, \Phi(Z,u)
    \, \tau^{(X)}_{u,t}
    \, B
    \\&=
    K \, \tau^{(X)}_{s,t} \, B
    + \calQ_1 + \calR_1
    \\&=
    \calQ_1 + \calR_1
    .
\end{align*}
In the third equality, we split the sum into sets intersecting~\(Y\) and sets disjoint from~\(Y\), which give~\(\calQ_1\) and~\(\calR_1\), respectively.
The last equality follows from \cref{lem:vanishing-term} applied to the path~\((X)\).

Now let \(n>0\) and fix a path \(\gamma_{:n} = \gamma\in\SP_n(X)\) with \(\gamma\perp Y\).
Applying Duhamel's formula to compare \(\tau^{\gamma_{:n-1}}\) with \(\tau^{\gamma_{:n}}\) gives
\begin{align*}
    \tau^{\gamma_{:n-1}}_{u_n,t} \, B
    &=
    \tau^{\gamma_{:n}}_{u_n,t} \, B
    +
    \intlap_{u_n}^t \odif{u_{n+1}}
    \, \tau^{\gamma_{:n-1}}_{u_n,u_{n+1}}
    \, \sumstack[lr]{Z\in \powerset(\Gamma)}
    \, \paren[\big]{
        \Phi^{\gamma_{:n-1}}(Z,u_{n+1}) - \Phi^{\gamma_{:n}}(Z,u_{n+1})
    }
    \, \tau^{\gamma_{:n}}_{u_{n+1},t} \, B
    \\&=
    \tau^{\gamma_{:n}}_{u_n,t} \, B
    +
    \intlap_{u_n}^t \odif{u_{n+1}} \, \tau^{\gamma_{:n-1}}_{u_n,u_{n+1}}
    \, \sumstack[lr]{
        Z\in \powerset(\Gamma)\suchthat\\
        \gamma_{:n-1} \perp Z \\
        \gamma_n \cap Z \neq \emptyset
    }
    \, \Phi(Z,u_{n+1}) \, \tau^{\gamma_{:n}}_{u_{n+1},t} \, B
    .
\end{align*}
The above sum runs over exactly the sets \(Z\) with \(\gamma \oplus Z \in \SP_{n+1}(X)\).
Note that it is also possible that \(\gamma\) cannot be extended without losing the self-avoiding property (cf.~\cref{subfig:self-avoiding-path-that-cannot-be-extended}).
In this case, the sum is automatically \(0\), and the following calculation shows that the first term also does not contribute to~\(\calR_n\).
Splitting the sum depending on whether~\(Z\) intersects~\(Y\), one obtains
\begin{equation*}
    \tau^{\gamma_{:n-1}}_{u_n,t} \, B
    =
    \tau^{\gamma_{:n}}_{u_n,t} \, B
    \begin{aligned}[t]
        &+
        \int_{u_n}^t \odif{u_{n+1}} \, \tau^{\gamma_{:n-1}}_{u_n,u_{n+1}}
        \quad \sumstack[lr]{
            Z\in \powerset(\Gamma)\suchthat\\
            \gamma\oplus Z \oplus Y \in \SP_{n+1}(X,Y)
        }
        \, \Phi(Z,u_{n+1}) \, \tau^{\gamma_{:n}}_{u_{n+1},t} \, B
        \\&+
        \int_{u_n}^t \odif{u_{n+1}} \, \tau^{\gamma_{:n-1}}_{u_n,u_{n+1}}
        \quad \sumstack[lr]{
            Z\in \powerset(\Gamma)\suchthat\\
            \gamma\oplus Z \in \SP_{n+1}(X)\\
            \gamma \oplus Z \perp Y
        }
        \, \Phi(Z,u_{n+1}) \, \tau^{\gamma_{:n}}_{u_{n+1},t} \, B
        .
    \end{aligned}
\end{equation*}
Inserting this identity into the definition of \(\calR_n\) results in
\begin{equation*}
    \calR_n
    =
    \begin{aligned}[t]
        &
        K
        \, \intlap_{\Delta^n_{s,t}} \odif{u}
        \, \sumstack{
            \gamma \in \SP_n(X) \suchthat\\
            \gamma \perp Y
        }
        \, \tau^{\gamma_{:-1}}_{s,u_1}
        \, \Phi(\gamma_1,u_1)
        \, \tau^{\gamma_{:0}}_{u_1,u_2}
        \, \Phi(\gamma_2,u_2)
        \, \dotsm
        \, \tau^{\gamma_{:n-2}}_{u_{n-1},u_n}
        \, \Phi(\gamma_{n},u_n)
        \, \tau^{\gamma_{:n}}_{u_n,t}
        \, B
        \\&+
        \calQ_{n+1}
        + \calR_{n+1}
        .
    \end{aligned}
\end{equation*}
The first term vanishes, since \cref{lem:vanishing-term} gives
\(
    \Phi(\gamma_n,u_n)
    \, \tau^{\gamma_{:n}}_{u_n,t} \, B
    =
    0
\).
Thus, \(\calR_n=\calQ_{n+1}+\calR_{n+1}\) and the induction is complete.

\begin{lemma}
    \label{lem:vanishing-term}
    Let \(\Phi\) be a finitely supported interaction, \(n\in \N_0\), \(\gamma \in \SP\) with \(\ell_\gamma=n\) and \(Y\in \powerset(\Gamma)\) with \(\gamma \perp Y\).
    Define \(\Phi^{\gamma}\) as in~\eqref{eq:thm-path-dependent-interactions} and let \(\tau^{\gamma}\) be its generated evolution.
    Then for all \(X \subseteq \bigcup_{j=0}^{\ell_\gamma} \gamma_j\), \(K\in \calB_X\), \(B\in \alg_Y\) and \((s,t)\in \Delta^2_I\), it holds that \(K \, \tau^{\gamma}_{s,t} \, B = 0\).
\end{lemma}

\begin{proof}
    Fix \(\Lambda = \bigcup_{j=0}^{\ell_\gamma} \gamma_j\), \(X\subseteq \Lambda\), \(K\in \calB_X\), \(B\in \alg_Y\) and \((s,t)\in \Delta^2_I\).
    We recall that \(\Phi^\gamma(Z,t)\) vanishes unless \(Z\subseteq \Gamma\setminus \Lambda\).
    Moreover, by assumption \(Y\subseteq \Gamma \setminus \Lambda\).
    Hence, looking at the definition of \(\tau^\gamma\), we find
    \begin{equation*}
        \tau^\gamma_{s,t} \, B
        =
        \sum_{m=0}^\infty
        \intlap_{\Delta^m_{s,t}} \odif{u}
        \, \sumstack[r]{Z_1\in \powerset(\Gamma)}
        \, \Phi^\gamma(Z_1,u_1)
        \, \dotsm
        \, \sumstack[lr]{Z_m\in \powerset(\Gamma)}
        \, \Phi^\gamma(Z_m,u_m)
        \, B
        \in
        \alg_{\Gamma\setminus \Lambda}
        .
    \end{equation*}
    Therefore, \(K \, \tau^\gamma_{s,t} \, B = 0\), since \(K\in \calB_{\Lambda}\).
\end{proof}

\minisec{Convergence of the expansion}
We now show that, as \(n\to\infty\) in~\eqref{eq:proof-finite-volume-truncated-expansion}, the sum over the \(\calQ_m\) converges and \(\calR_n\) goes to zero.
We first establish a path estimate that controls these terms.
It will also be used to prove existence of the thermodynamic limit in \cref{sec:proof-thermodynamic-limit}.

\begin{lemma}[Summability of path weights]
    \label{lem:summability-of-path-weights}
    Let \(\Phi\) be an interaction satisfying~\eqref{eq:weighted-interaction-summability}.
    For all \(X\in\powerset(\Gamma)\), \(n\in\N_{>0}\), and \((u_1,\dotsc,u_n)\in I^n\), we have
    \begin{equation}
        \label{eq:lem-summability-of-path-weights}
        \sumstack{\gamma\in\Path_n(X)}
        \, \prod_{j=1}^n \, \norm{\Phi(\gamma_j,u_j)}
        \leq
        \abs{X} \, \tnorm{\Phi}_1 \, C_\Phi^{n-1}
        .
    \end{equation}
\end{lemma}

\begin{proof}
    First, note that for \(X=\emptyset\) there is no path \(\gamma\in \Path_n(\emptyset)\), and both sides vanish.
    Hence, we assume \(\abs{X} \geq 1\) in the following.
    Also note that \(\tnorm{\Phi}_1 \leq C_\Phi\) is finite.
    We now prove~\eqref{eq:lem-summability-of-path-weights} by induction over~\(n\in \N_{>0}\).
    For \(n=1\), every path in \(\Path_1(X)\) consists of~\(X\) and one set~\(Z\) intersecting~\(X\).
    Thus,
    \begin{equation*}
        \sumstack[r]{\gamma\in\Path_1(X)}
        \, \norm{\Phi(\gamma_1,u_1)}
        \leq
        \sum_{x\in X}
        \, \sumstack[r]{Z\in\powerset(\Gamma)\\ x\in Z}
        \, \norm{\Phi(Z,u_1)}
        \leq
        \abs{X} \, \tnorm{\Phi}_1
        .
    \end{equation*}
    Assume that~\eqref{eq:lem-summability-of-path-weights} holds for some~\(n\in \N_{>0}\).
    Fix \(X\in\powerset(\Gamma)\) and \((u_1,\dotsc,u_{n+1})\in I^{n+1}\).
    Every path \(\gamma\in\Path_{n+1}(X)\) is determined by its first intermediate set \(Z=\gamma_1\), and the remaining path \(\tilde{\gamma}=(Z,\gamma_2,\dotsc,\gamma_{n+1})\in\Path_n(Z)\).
    It follows that 
    \begin{align*}
        \sumstack{\gamma\in\Path_{n+1}(X)}
        \, \prod_{j=1}^{n+1} \, \norm{\Phi(\gamma_j,u_j)}
        &\leq
        \sumstack[lr]{
            Z\in\powerset(\Gamma)\\
            Z\cap X\neq\emptyset
        }
        \, \norm{\Phi(Z,u_1)}
        \, \sumstack[l]{\tilde{\gamma} \in\Path_n(Z)}
        \, \prod_{j=1}^n \, \norm{\Phi(\tilde{\gamma}_j,u_{j+1})}
        \\&\leq
        \sumstack[lr]{
            Z\in\powerset(\Gamma)\\
            Z\cap X\neq\emptyset
        }
        \, \norm{\Phi(Z,u_1)}
        \, \abs{Z}
        \, \tnorm{\Phi}_1 \, C_\Phi^{n-1}
        \\&\leq
        \abs{X}
        \, \tnorm{\Phi}_1 \, C_\Phi^{n}
        ,
    \end{align*}
    where we used the induction hypothesis in the second inequality and~\eqref{eq:weighted-interaction-summability} in the third.
    This completes the induction.
\end{proof}

To apply this result, notice that~\(C_\Phi<\infty\) because~\(\Phi\) is finitely supported.
Contractivity of the evolutions and the triangle inequality yield \(\norm{\calQ_0}\),~\(\norm{\calR_0} \leq \norm{K} \, \norm{B}\) and for \(n>0\)
\begin{align*}
    \norm{\calQ_n}
    &\leq
    \norm{K} \, \norm{B}
    \, \intlap_{\Delta^n_{s,t}} \odif{u}
    \, \sumstack{\gamma\in\SP_n(X,Y)}
    \, \prod_{j=1}^n \, \norm{\Phi(\gamma_j,u_j)}
    ,
    \\
    \norm{\calR_n}
    &\leq
    \norm{K} \, \norm{B}
    \, \intlap_{\Delta^n_{s,t}} \odif{u}
    \, \sumstack{
        \gamma\in\SP_n(X)\\
        \gamma\perp Y
    }
    \, \prod_{j=1}^n \, \norm{\Phi(\gamma_j,u_j)}
    .
\end{align*}
The paths contributing to \(\calR_n\) form a subset of \(\SP_n(X)\), and hence of \(\Path_n(X)\).
For \(\calQ_n\), every path ends at the fixed set~\(Y\) and deleting this last set defines the injective map
\begin{equation*}
    \SP_n(X,Y) \to \Path_n(X)
    ,
    \qquad
    \gamma \mapsto \gamma_{:n}
    .
\end{equation*}
Thus, we can upper bound
\begin{align*}
    \norm{\calQ_n},\norm{\calR_n}
    &\leq
    \norm{K} \, \norm{B}
    \, \intlap_{\Delta^n_{s,t}} \odif{u}
    \, \sumstack{\gamma \in \Path_n(X)}
    \, \prod_{j=1}^n
    \, \norm{\Phi(\gamma_j,u_j)}
    \\&\leq
    \norm{K} \, \norm{B}
    \, \abs{X} \, \tnorm{\Phi}_1 \, C_\Phi^{n-1}
    \, \frac{(t-s)^n}{n!}
    ,
\end{align*}
where we applied \cref{lem:summability-of-path-weights} and used the simplex volume.

Hence,
\begin{equation*}
    \sum_{n=0}^\infty \, \norm{\calQ_n}
    \leq
    \norm{K} \, \norm{B}
    \, \abs{X}
    \, \e^{C_\Phi \, (t-s)}
    <
    \infty
\end{equation*}
and \(\norm{\calR_n} \to 0\), and therefore
\begin{equation*}
    K \, \tau_{s,t} \, B
    =
    \sum_{m=0}^\infty \, \calQ_m
    ,
\end{equation*}
with the series being absolutely convergent.
This concludes the proof of \cref{thm:adaptive-dyson-expansion} for finitely supported interactions.

\subsection{Taking the thermodynamic limit}
\label{sec:proof-thermodynamic-limit}

We first prove existence of the thermodynamic limit using the expansion for finitely supported interactions.

\begin{proposition}
    \label{prop:existence-of-the-thermodynamic-limit}
    Let \(\Phi\) be an interaction satisfying~\eqref{eq:weighted-interaction-summability}.
    For each \(\Lambda \in \powerset(\Gamma)\), consider the restricted interaction \(\Phi\rvert_\Lambda\) defined in~\eqref{eq:restricted-interaction} and let \(\tau^{\Lambda}\) be its generated evolution.
    Then for each \((s,t)\in \Delta^2_I\) and \(B\in \alg\), the net \((\tau_{s,t}^\Lambda \, B)_{\Lambda \in \powerset(\Gamma)}\) converges in norm.
\end{proposition}

\begin{proof}
    Fix \((s,t)\in\Delta_I^2\).
    First let \(B\in \alg_Y\) with \(Y\in \powerset(\Gamma)\).
    We show that the net \((\tau_{s,t}^\Lambda \, B)_{\Lambda \in \powerset(\Gamma)}\) is Cauchy.

    Let \(\Lambda\in\powerset(\Gamma)\) contain \(Y\), and let \(\Lambda_1\),~\(\Lambda_2\in\powerset(\Gamma)\) contain \(\Lambda\).
    Duhamel's formula gives
    \begin{equation*}
        \paren[\big]{\tau_{s,t}^{\Lambda_1}-\tau_{s,t}^{\Lambda_2}} \, B
        =
        \intlap_s^t \odif{u}
        \, \tau_{s,u}^{\Lambda_1}
        \, \paren[\bigg]{
            \sumstack{Z \subseteq \Lambda_1} \Phi(Z,u)
            - \sumstack{Z \subseteq \Lambda_2} \Phi(Z,u)
        }
        \, \tau_{u,t}^{\Lambda_2} \, B
        .
    \end{equation*}
    Every term with \(Z\subseteq \Lambda_1\cap \Lambda_2\) occurs in both sums and cancels, all others satisfy \(Z \not\subseteq \Lambda\).
    Thus, taking norms, using Tonelli and \cref{eq:thm-adaptive-dyson-expansion-norm-bound} for the finitely supported~\(\Phi\rvert_{\Lambda_2}\) to bound \(\norm{\Phi(Z,u) \, \tau_{u,t}^{\Lambda_2} \, B}\), and enlarging the sum over self-avoiding paths to paths yields
    \begin{align}
        \notag
        \norm[\big]{\paren[\big]{\tau_{s,t}^{\Lambda_1}-\tau_{s,t}^{\Lambda_2}} \, B}
        &\leq
        \, \sumstack[lr]{
            Z\in\powerset(\Gamma)\\
            Z\not\subseteq\Lambda
        }\quad
        \intlap_s^t \odif{u}
        \, \norm[\big]{\Phi(Z,u) \, \tau_{u,t}^{\Lambda_2} \, B}
        \\&\leq
        \label{eq:proof-TDL-expansion}
        \norm{B}
        \, \sumstack[lr]{
            Z\in\powerset(\Gamma)\\
            Z\not\subseteq\Lambda
        }
        \quad \intlap_s^t \odif{u}
        \, \norm{\Phi(Z,u)}
        \, \sum_{n=0}^\infty
        \, \intlap_{\Delta^n_{u,t}} \odif{v}
        \sumstack{\gamma\in\Path_n(Z,Y)}
        \, \prod_{j=1}^n \, \norm{\Phi(\gamma_j,v_j)}
        .
    \end{align}
    The bound is independent of \(\Lambda_1\) and \(\Lambda_2\).
    Moreover, the outer sum on the right-hand side of~\eqref{eq:proof-TDL-expansion} without the condition \(Z\not\subseteq\Lambda\) converges absolutely.
    To see this, consider the bijection
    \begin{align*}
        \Path_{n+1}(Y)
        &\to
        \Set{ (Z,\gamma) \given Z \in \powerset(\Gamma),\ \gamma\in \Path_n(Z,Y)}
        \\
        \eta
        &\mapsto
        \paren[\big]{
            \eta_{n+1},
            (\eta_{n+1}, \eta_n, \dotsc, \eta_{1}, \eta_{0} = Y)
        }
        .
    \end{align*}
    Therefore, after using Tonelli’s theorem, \cref{lem:summability-of-path-weights} and  \(
        \int_s^t \odif{u} \int_{\Delta^n_{u,t}} \odif{v}
        =
        \int_{\Delta^{n+1}_{s,t}} \odif{v}
    \), we can upper bound
    \begin{align*}
        \Alignindent
        \norm{B}
        \, \sumstack[lr]{
            Z\in\powerset(\Gamma)
        }
        \quad \intlap_s^t \odif{u}
        \, \norm{\Phi(Z,u)}
        \, \sum_{n=0}^\infty
        \, \intlap_{\Delta^n_{u,t}} \odif{v}
        \sumstack{\gamma\in\Path_n(Z,Y)}
        \, \prod_{j=1}^n \, \norm{\Phi(\gamma_j,v_j)}\\
        &\leq
        \norm{B}
        \, \sum_{n=0}^\infty
        \, \intlap_s^t \odif{u}
        \intlap_{\Delta^n_{u,t}} \odif{v}
        \sumstack[r]{\eta\in\Path_{n+1}(Y)}
        \, \norm{\Phi(\eta_{n+1},u)}
        \, \prod_{j=1}^n \, \norm{\Phi(\eta_j,v_{n+1-j})}
        \\&\leq
        \norm{B}
        \, \sum_{n=0}^\infty
        \, \frac{(t-s)^{n+1}}{(n+1)!}
        \, \abs{Y}
        \, \tnorm{\Phi}_1
        \, C_\Phi^{n}
        \\&\leq
        \norm{B}
        \, \abs{Y}
        \, \e^{C_\Phi \, (t-s)}
        .
    \end{align*}

    This, together with the fact that each finite set \(Z\) is eventually contained in \(\Lambda\) as~\(\Lambda\nearrow\Gamma\), implies that the right-hand side of~\eqref{eq:proof-TDL-expansion} tends to zero as \(\Lambda\nearrow\Gamma\).
    Since \(\Lambda\subset \Lambda_1,\Lambda_2\), this implies \(
        \norm[\big]{\paren[\big]{\tau_{s,t}^{\Lambda_1}-\tau_{s,t}^{\Lambda_2}} \, B} \to 0
    \) as \(\Lambda_1\),~\(\Lambda_2 \nearrow \Gamma\). 

    Now let \(B\in \alg\), and let \((B_n)_{n\in \N_0}\) be a sequence of operators supported on finite sets approximating~\(B\).
    We find for all \(n\in \N_0\) that
    \begin{align*}
        \norm{\tau_{s,t}^{\Lambda_1} \, B - \tau_{s,t}^{\Lambda_2} \, B}
        &\leq
        \norm{\tau_{s,t}^{\Lambda_1} \, B - \tau_{s,t}^{\Lambda_1} \, B_n}
        +
        \norm{\tau_{s,t}^{\Lambda_1} \, B_n - \tau_{s,t}^{\Lambda_2} \, B_n}
        +
        \norm{\tau_{s,t}^{\Lambda_2} \, B_n - \tau_{s,t}^{\Lambda_2} \, B}
        \\&\leq
        \norm{B-B_n}
        +
        \norm{\tau_{s,t}^{\Lambda_1} \, B_n - \tau_{s,t}^{\Lambda_2} \, B_n}
        +
        \norm{B_n - B}
        .
    \end{align*}
    Choosing first \(n\) and then \(\Lambda_1\),~\(\Lambda_2\) sufficiently large shows that the net is Cauchy and thus converges in norm.
\end{proof}

We now complete the proof of \cref{thm:adaptive-dyson-expansion} by passing to the thermodynamic limit in the expansion.
Let \(\Phi\) be an interaction satisfying~\eqref{eq:weighted-interaction-summability}.
By \cref{prop:existence-of-the-thermodynamic-limit}, \(\Phi\) generates a time evolution \(\tau\).
Fix \((s,t)\in\Delta_I^2\), \(X\),~\(Y\in\powerset(\Gamma)\), \(K\in\calB_X\), and \(B\in\alg_Y\).
For every \(\eta \in\SP\), the interaction \(\Phi^{\eta}\) defined in~\eqref{eq:thm-path-dependent-interactions} is obtained from \(\Phi\) by setting some terms to zero.
It therefore also satisfies~\eqref{eq:weighted-interaction-summability}, so \cref{prop:existence-of-the-thermodynamic-limit} gives a time evolution \(\tau^{\eta}\) generated by \(\Phi^{\eta}\).

We first show that the right-hand side of~\eqref{eq:thm-adaptive-dyson-expansion} is well-defined. 
Fix \(n\in\N_{>0}\),~\(\gamma\in\SP_n(X,Y)\), and \(u=(u_1,\dotsc,u_n)\in\Delta^n_{s,t}\).
All evolution factors are contractive, so the norm of each integrand in~\eqref{eq:thm-adaptive-dyson-expansion} satisfies
\begin{equation*}
    \norm[\big]{
        K
        \, \tau_{s,u_1}
        \, \Phi(\gamma_1,u_1)
        \, \tau^{\gamma_{:0}}_{u_1,u_2}
        \, \dotsm
        \, \Phi(\gamma_n,u_n)
        \, \tau^{\gamma_{:n-1}}_{u_n,t}
        \, B
    }
    \leq
    \norm{K} \, \norm{B}
    \, \prod_{j=1}^n \, \norm{\Phi(\gamma_j,u_j)}
    .
\end{equation*}
As in the finite-support argument, deleting the final set \(Y\) defines an injective map from \(\SP_n(X,Y)\) to \(\Path_n(X)\) and preserves this product.
Applying \cref{lem:summability-of-path-weights} and integrating over the simplex therefore gives
\begin{equation}
    \label{eq:proof-absolute-convergence-bound}
    \begin{aligned}
        \sum_{n=1}^\infty
        \, \intlap_{\Delta^n_{s,t}} \odif{u}
        \, \sumstack{\gamma\in\SP_n(X,Y)}
        \, \prod_{j=1}^n \, \norm{\Phi(\gamma_j,u_j)}
        &\leq
        \abs{X}
        \, \sum_{n=1}^\infty
        \, \frac{(t-s)^n}{n!}
        \, \tnorm{\Phi}_1
        \, C_\Phi^{n-1}
        \\&\leq
        \abs{X} \, \e^{C_\Phi \, (t-s)}
        .
    \end{aligned}
\end{equation}
For \(n=0\), the sum on the right-hand side of~\eqref{eq:thm-adaptive-dyson-expansion} contains at most one term, whose norm is bounded by \(\norm{K} \, \norm{B}\).
Furthermore, because of \cref{def:time-evolution-general-case} and \cref{prop:finite-evolutions}, the integrand on the right-hand side of~\eqref{eq:thm-adaptive-dyson-expansion} is a pointwise norm limit of continuous restricted integrands and so strongly measurable.
Hence, the time integrals appearing in~\eqref{eq:thm-adaptive-dyson-expansion} exist as Bochner integrals, and the right-hand side of~\eqref{eq:thm-adaptive-dyson-expansion} converges absolutely in norm.

We now consider the net of restricted interactions \((\Phi\rvert_\Lambda)_{\Lambda\in \powerset(\Gamma)}\) as defined in~\eqref{eq:restricted-interaction} and for each \(\Lambda \in \powerset(\Gamma)\), \(\gamma \in \SP\) we define the interaction \((\Phi\rvert_\Lambda)^{\gamma}\) as in~\eqref{eq:thm-path-dependent-interactions} and denote by \(\tau^{\Lambda,\gamma}\) its generated evolution.
By the statement for finitely supported interactions proven in \cref{sec:proof-finite-volume}, we have for all \(\Lambda \in \powerset(\Gamma)\)
\begin{equation}
    \label{eq:expansion-of-restricted-evolution}
    K \, \tau^\Lambda_{s,t} \, B
    =
    K
    \, \sum_{n=0}^\infty
    \, \intlap_{\Delta^n_{s,t}} \odif{u}
    \, \sumstack[r]{\gamma\in \SP_n(X,Y)}
    \,
    \begin{aligned}[t]
        \tau^{\Lambda,\gamma_{:-1}}_{s,u_1}
        \, & \Phi\rvert_\Lambda(\gamma_{1},u_1)
        \, \tau^{\Lambda,\gamma_{:0}}_{u_1,u_2}
        \, \Phi\rvert_\Lambda(\gamma_2,u_2)
        \\&\dotsm
        \, \tau^{\Lambda,\gamma_{:n-2}}_{u_{n-1},u_n}
        \, \Phi\rvert_\Lambda(\gamma_{n},u_n)
        \, \tau^{\Lambda,\gamma_{:n-1}}_{u_n,t}
        \, B
        .
    \end{aligned}
\end{equation}
The left-hand side of~\eqref{eq:expansion-of-restricted-evolution} converges to \(K \, \tau_{s,t} \, B\) due to \cref{prop:existence-of-the-thermodynamic-limit}.
Notice that for every path~\(\eta\), it holds that \(
    (\Phi\rvert_\Lambda)^{\eta}
    =
    (\Phi^{\eta})\rvert_\Lambda
\) and hence by \cref{prop:existence-of-the-thermodynamic-limit}
\begin{equation*}
    \tau^{\Lambda,\eta}_{u,v} \, A
    \to
    \tau^{\eta}_{u,v} \, A
    \qquadtext{as}
    \Lambda\nearrow\Gamma
\end{equation*}
in norm for all \(A\in\alg\) and \((u,v)\in\Delta_I^2\).
Thus, every integrand in~\eqref{eq:expansion-of-restricted-evolution} converges pointwise to the corresponding integrand in~\eqref{eq:thm-adaptive-dyson-expansion}.
Moreover, the bound~\eqref{eq:proof-absolute-convergence-bound} also applies to the restricted expansion~\eqref{eq:expansion-of-restricted-evolution}.
Therefore, by dominated convergence the right-hand side of~\eqref{eq:expansion-of-restricted-evolution} also converges to the right-hand side of~\eqref{eq:thm-adaptive-dyson-expansion} along any sequence~\(\Lambda_m \nearrow \Gamma\).
This concludes the proof of \cref{thm:adaptive-dyson-expansion}.

Finally, consider~\eqref{eq:thm-adaptive-dyson-expansion} and suppose that \(X\cap Y=\emptyset\).
Then \(\SP_0(X,Y)=\emptyset\), and the \(n=0\) term vanishes.
Applying the triangle inequality and using contractivity of the evolutions gives~\eqref{eq:thm-adaptive-dyson-expansion-norm-bound}.

\section{Proofs of Lieb-Robinson bounds}
\label{sec:proofs-Lieb-Robinson-bounds}

Within this section we collect the proofs of the various Lieb-Robinson bounds stated in \cref{sec:results}.

\subsection{Proof of \texorpdfstring{\cref{thm:Lieb-Robinson-bound}}{Theorem~\ref*{thm:Lieb-Robinson-bound}}}
\label{sec:proof-Lieb-Robinson-bound}

Since \(\Phi\) has \(F\)-decay, it satisfies~\eqref{eq:weighted-interaction-summability}, so \cref{thm:adaptive-dyson-expansion} applies.
Fix \((s,t)\in\Delta_I^2\), disjoint \(X\),~\(Y\in\powerset(\Gamma)\), \(K\in\calB_X\), and \(B\in\alg_Y\).

To obtain~\eqref{eq:thm-Lieb-Robinson-bound}, observe that every intermediate set of a self-avoiding path \(\gamma\in \SP_n(X,Y)\) contains at least two sites and is not contained entirely in \(X\) or \(Y\).
Thus, we can replace \(\Phi\) by \(\tilde{\Phi}\) in~\eqref{eq:thm-adaptive-dyson-expansion-norm-bound} and then overcount by summing over all paths instead of only self-avoiding ones.
Therefore, inequality~\eqref{eq:thm-Lieb-Robinson-bound} follows from
\begin{align*}
    \norm{K \, \tau_{s,t} \, B}
    &\leq
    \norm{K} \, \norm{B}
    \, \sum_{n=1}^{\infty}
    \, \intlap_{\Delta^n_{s,t}} \odif{u}
    \, \sumstack{\gamma\in\Path_n(X,Y)}
    \, \prod_{j=1}^n
    \, \norm{\tilde{\Phi}(\gamma_j,u_j)}
    \\&\leq
    \norm{K} \, \norm{B}
    \, \sum_{n=1}^{\infty}
    \, \sum_{x\in X} \, \sum_{y\in Y} \, F\pdist{x,y}
    \, \frac{C_F^{n-1} \, I_{\tilde{\Phi}}(s,t)^n}{n!}
    ,
\end{align*}
where in the last step we have applied \cref{lem:bounding-sum-over-n-long-paths}, given after the proof.

\begin{lemma}
    \label{lem:bounding-sum-over-n-long-paths}
    Let \(F\) be a decay function on \(\Gamma\), and let \(\Phi\) be an interaction with \(F\)-decay.
    Fix \((s,t)\in \Delta^2_I\).
    For all \(X\),~\(Y\in \powerset(\Gamma)\) and \(n\in \N_{>0}\), it holds that
    \begin{align}
        \label{eq:lem-bounding-sum-over-n-long-paths-from-X-to-Y}
        \intlap_{\Delta^n_{s,t}} \odif{u}
        \, \sumstack{\gamma\in\Path_n(X,Y)}
        \, \prod_{j=1}^n
        \, \norm{\Phi(\gamma_j,u_j)}
        &\leq
        \sum_{x\in X} \, \sum_{y\in Y} \, F\pdist{x,y}
        \, \frac{C_F^{n-1} \, I_{\Phi}(s,t)^n}{n!}
    \shortintertext{and}
        \label{eq:lem-bounding-sum-over-n-long-paths-attached-to-X}
        \intlap_{\Delta^n_{s,t}} \odif{u}
        \, \sumstack{\gamma\in\Path_n(X)}
        \, \prod_{j=1}^n
        \, \norm{\Phi(\gamma_j,u_j)}
        &\leq
        \abs{X} \, \norm{F}
        \, \frac{C_F^{n-1} \, I_{\Phi}(s,t)^n}{n!}
        .
    \end{align}
    For \(n=0\), the left-hand side of~\eqref{eq:lem-bounding-sum-over-n-long-paths-from-X-to-Y} equals \(1\) if \(X\cap Y\neq\emptyset\) and \(0\) otherwise, while that of~\eqref{eq:lem-bounding-sum-over-n-long-paths-attached-to-X} equals \(1\).
\end{lemma}

\begin{proof}
    For \(n=0\), there is exactly one path in \(\Path_0(X)\), namely \((X)\).
    On the other hand, \(\Path_0(X,Y)\) consists of \((X,Y)\) if \(X\cap Y\neq\emptyset\) and is empty otherwise.
    The claimed values follow, since the empty product and the integral over \(\Delta^0_{s,t}\) equal \(1\).
    Therefore, in the remainder of the proof we assume \(n\in\N_{>0}\).

    For every path \(\gamma\), we can find a point in each intersection \(z_j\in \gamma_{j-1} \cap \gamma_{j}\).
    Hence, we can upper bound the sum over paths~\eqref{eq:lem-bounding-sum-over-n-long-paths-from-X-to-Y} as
    \begin{align*}
        \intlap_{\Delta^n_{s,t}} \odif{u}
        \, \sumstack{\gamma\in\Path_n(X,Y)}
        \, \prod_{j=1}^n
        \, \norm{\Phi(\gamma_j,u_j)}
        &\leq
        \intlap_{\Delta^n_{s,t}} \odif{u}
        \, \sumstack[l]{
            z_1\in X\\
            z_2,\dotsc,z_{n} \in\Gamma\\
            z_{n+1} \in Y
        }
        \, \prod_{j=1}^n
        \, \sumstack[r]{
            \gamma_j\in \powerset(\Gamma)\suchthat\\
            z_j,z_{j+1} \in \gamma_j
        }
        \, \norm{\Phi(\gamma_j,u_j)}
        \\&\leq
        \intlap_{\Delta^n_{s,t}} \odif{u}
        \, \sumstack[l]{
            z_1\in X\\
            z_2,\dotsc,z_{n} \in\Gamma\\
            z_{n+1} \in Y
        }
        \, \prod_{j=1}^n
        \, F\pdist{z_j,z_{j+1}}
        \, \norm{\Phi(u_j)}_F
        \\&\leq
        \intlap_{\Delta^n_{s,t}} \odif{u}
        \, \sumstack[lr]{
            z_1\in X\\
            z_{n+1} \in Y
        }
        \, C_F^{n-1}
        \, F\pdist{z_1,z_{n+1}}
        \, \prod_{j=1}^n
        \, \norm{\Phi(u_j)}_F
        \\&=
        \sumstack[lr]{
            z_1\in X\\
            z_{n+1} \in Y
        }
        \, C_F^{n-1}
        \, F\pdist{z_1,z_{n+1}}
        \, \frac{1}{n!}
        \, \paren[\bigg]{\int_s^t \odif{u} \, \norm{\Phi(u)}_F}^n
        ,
    \end{align*}
    where we used the definition~\eqref{eq:interaction-norm} of the interaction norm in the second step, the convolution property~\eqref{eq:convolution-constant} \(n-1\) times in the third step, and the symmetry of the integrand under permutations of the time variables in the last step.

    To obtain~\eqref{eq:lem-bounding-sum-over-n-long-paths-attached-to-X}, notice that every path \(\gamma\in \Path_n(X)\) has a final non-empty set containing a point \(y\in \Gamma\).
    Therefore, by overcounting
    \begin{align*}
        \intlap_{\Delta^n_{s,t}} \odif{u}
        \, \sumstack{\gamma\in\Path_n(X)}
        \, \prod_{j=1}^n
        \, \norm{\Phi(\gamma_j,u_j)}
        &\leq
        \intlap_{\Delta^n_{s,t}} \odif{u}
        \, \sum_{y\in \Gamma}
        \, \sumstack{\gamma\in\Path_n(X,\List{y})}
        \, \prod_{j=1}^n
        \, \norm{\Phi(\gamma_j,u_j)}
        \\&\leq
        \intlap_{\Delta^n_{s,t}} \odif{u}
        \, \sum_{x\in X} \, \sum_{y\in \Gamma}
        \, C_F^{n-1}
        \, F\pdist{x,y}
        \, \prod_{j=1}^n
        \, \norm{\Phi(u_j)}_F
        \\&\leq
        \abs{X} \, \norm{F}
        \, \frac{C_F^{n-1} \, I_{\Phi}(s,t)^n}{n!}
        ,
    \end{align*}
    where we used the bound on the integrand from the first three steps of the previous estimate, applied with \(Y=\List{y}\), in the second step.
    This concludes the proof.
\end{proof}

\subsection{Proof of \texorpdfstring{\cref{cor:Lieb-Robinson-bound-nearest-neighbour-spin-chain}}{Corollary~\ref*{cor:Lieb-Robinson-bound-nearest-neighbour-spin-chain}}}
\label{sec:proof-Lieb-Robinson-bound-nearest-neighbour-spin-chain}

Since \(\Phi\) is a nearest-neighbour interaction on \(\Z\), \(C_\Phi \leq 2 \, \tnorm{\Phi}_1<\infty\).
Therefore, \cref{thm:adaptive-dyson-expansion} applies and gives the bound~\eqref{eq:thm-adaptive-dyson-expansion-norm-bound}.

Notice that, due to the assumptions on the interactions and self-avoidance, the paths contributing to the sum in~\eqref{eq:thm-adaptive-dyson-expansion-norm-bound} connect \(X\) to \(Y\) through sets of two neighbouring sites, with consecutive intermediate sets intersecting in one site.
Since \(X\) and \(Y\) are ordered and disjoint, there is exactly one such path between~\(X\) and \(Y\).

Let us consider the case \(X<Y\) and set \(x=\max X\) and \(y=\min Y\).
The other case is analogous.
Then, the path is given by
\begin{equation*}
    \gamma
    =
    \paren[\big]{X,\List{x,x+1},\dotsc,\List{y-1,y}, Y}
    \in
    \SP_{\dist{X,Y}}(X,Y)
    .
\end{equation*}
Therefore, we obtain
\begin{align*}
    \norm{K \, \tau_{s,t} \, B}
    &\leq
    \norm{K} \, \norm{B}
    \, \intlap_{\Delta^{\dist{X,Y}}_{s,t}} \odif{u}
    \ \prod_{j=1}^{\dist{X,Y}}
    \, \norm[\big]{\Phi(\List{x+j-1,x+j},u_j)}
    \\&\leq
    \norm{K} \, \norm{B}
    \, \frac{
        \paren[\big]{
            (t-s)
            \, \sup_{u\in \intervalcc{s,t}}
            \, \sup_{z\in \List{x,\dotsc, y -1}}
            \, \norm[\big]{\Phi(\List{z,z+1}, u)}
        }^{\dist{X,Y}}
    }{
        \dist{X,Y}!
    }
    ,
\end{align*}
which concludes the proof.

\subsection{Proof of \texorpdfstring{\cref{thm:finite-range-Lieb-Robinson-bound}}{Theorem~\ref*{thm:finite-range-Lieb-Robinson-bound}}}
\label{sec:proof-finite-range-Lieb-Robinson-bound}

Since \(\Phi\) is \(k\)-body and \(\tnorm{\Phi}_1<\infty\), condition~\eqref{eq:weighted-interaction-summability} is satisfied.
\Cref{thm:adaptive-dyson-expansion} therefore implies that \(\Phi\) generates an evolution~\(\tau\) which satisfies the bound~\eqref{eq:thm-adaptive-dyson-expansion-norm-bound}.

As discussed in the proof of \cref{thm:Lieb-Robinson-bound}, every intermediate set of a self-avoiding path \(\gamma\) from~\(X\) to~\(Y\) contains at least two sites and is not contained entirely in~\(X\) or~\(Y\).
Therefore, we may replace \(\Phi\) by \(\tilde{\Phi}\) in the interaction factors of~\eqref{eq:thm-adaptive-dyson-expansion-norm-bound}.
We assume that \(\tnorm{\tilde{\Phi}}_1>0\), since otherwise all terms in the expansion~\eqref{eq:thm-adaptive-dyson-expansion} vanish and the claim follows.

Because \(\tilde{\Phi}\) has range \(R\), every intermediate set of a path contributing to the sum in~\eqref{eq:thm-adaptive-dyson-expansion-norm-bound} has diameter at most \(R\).
Consecutive sets intersect, so the triangle inequality gives \(\dist{X,Y} \leq n \, R\), where \(n\) is the number of intermediate sets of such a path.
This restriction on \(n\) together with~\eqref{eq:thm-adaptive-dyson-expansion-norm-bound} implies
\begin{equation}
    \label{eq:proof-finite-range-self-avoiding-paths}
    \norm{K \, \tau_{s,t} \, B}
    \leq
    \norm{K} \, \norm{B}
    \, \sum_{n=\ceil{\dist{X,Y}/R}}^{\infty}
    \, \intlap_{\Delta^n_{s,t}} \odif{u}
    \, \sumstack{\gamma\in\SP_n(X,Y)}
    \, \prod_{j=1}^n
    \, \norm{\tilde{\Phi}(\gamma_j,u_j)}
    .
\end{equation}
Enlarging the sum to all paths starting at~\(X\) leads to
\begin{equation}
    \label{eq:proof-finite-range-paths-from-X}
    \norm{K \, \tau_{s,t} \, B}
    \leq
    \norm{K} \, \norm{B}
    \, \sum_{n=\ceil{\dist{X,Y}/R}}^{\infty}
    \, \intlap_{\Delta^n_{s,t}} \odif{u}
    \, \sumstack{\gamma\in\Path_n(X)}
    \, \prod_{j=1}^n
    \, \norm{\tilde{\Phi}(\gamma_j,u_j)}
    .
\end{equation}
Reversing a path \(\paren[\big]{X,\gamma_1,\dotsc,\gamma_n,Y} \in\SP_n(X,Y)\) gives a path \(\paren[\big]{Y,\gamma_n,\dotsc,\gamma_1,X} \in\SP_n(Y,X)\).
Returning to~\eqref{eq:proof-finite-range-self-avoiding-paths}, we reverse the paths and enlarge the sum to all paths starting at~\(Y\) to obtain
\begin{equation}
    \label{eq:proof-finite-range-paths-from-Y}
    \norm{K \, \tau_{s,t} \, B}
    \leq
    \norm{K} \, \norm{B}
    \, \sum_{n=\ceil{\dist{X,Y}/R}}^{\infty}
    \, \intlap_{\Delta^n_{s,t}} \odif{u}
    \, \sumstack{\gamma\in\Path_n(Y)}
    \, \prod_{j=1}^n
    \, \norm{\tilde{\Phi}(\gamma_j,u_{n+1-j})}
    .
\end{equation}
Since \(\tilde{\Phi}\) is \(k\)-body, estimate~\eqref{eq:lem-summability-of-path-weights} in \cref{lem:summability-of-path-weights} applies to the sums in~\eqref{eq:proof-finite-range-paths-from-X} and~\eqref{eq:proof-finite-range-paths-from-Y} with \(C_{\tilde{\Phi}} \leq k \, \tnorm{\tilde{\Phi}}_1\).
Thus, after integrating over the simplex, we obtain
\begin{align}
    \label{eq:proof-finite-range-after-path-summability}
    \norm{K \, \tau_{s,t} \, B}
    &\leq
    \norm{K} \, \norm{B}
    \, \min\List{\abs{X},\abs{Y}}
    \, \sum_{n=\ceil{\dist{X,Y}/R}}^{\infty}
    \, \frac{k^{n-1} \, \tnorm{\tilde{\Phi}}_1^n \, \paren{t-s}^n}{n!}
    \\&\leq
    \norm{K} \, \norm{B}
    \, \min\List{\abs{X},\abs{Y}}
    \, k^{-1}
    \, \sum_{n=\ceil{\dist{X,Y}/R}}^{\infty}
    \, \frac{\paren[\big]{k \, \tnorm{\tilde{\Phi}}_1 \, \paren{t-s}}^n}{n!}
    \, \e^{n-\dist{X,Y}/R}
    \notag
    .
\end{align}
In the second step, we multiplied each summand by \(\e^{n-\dist{X,Y}/R} \geq 1\).
Bounding the series with the exponential proves the statement.

\subsection{Proof of \texorpdfstring{\cref{cor:localization-of-operators}}{Corollary~\ref*{cor:localization-of-operators}}}
\label{sec:proof-localization-of-operators}

For the proof we use a conditional expectation, which corresponds to the normalized partial trace for finite~\(\Gamma\) and finite-dimensional local Hilbert spaces~\(\HS_x\).
It allows localizing any \(A\in \alg\) in a finite region \(X\in \powerset(\Gamma)\) up to an explicit error.

\begin{lemma}[{\cite[Corollary~4.4, Proposition~4.5]{NSY2019}}]
    \label{lem:conditional-expectation-spins}
    Let \(X\in \powerset(\Gamma)\).
    Then there exists a unit-preserving, completely positive linear map \(\cexpsym_{X} \colon \alg\to\alg\) satisfying, among other things,
    \begin{thmlist}
        \item
            \(\cexpsym_{X} \, \alg \subseteq \alg_X\),
        \item
            \label{lem:cmean-spins-E(ABC)=AE(B)C}
            \(\cexpsym_{X} \, A=A\) for all \(A\in\alg_{X}\),
        \item
            \label{lem:cmean-spins-unit-norm}
            \(\norm{\cexpsym_{X}}=1\),
        \item
            \label{lem:cmean-spins-bound-for-difference-A-E(A)}
            If for some \(B\in\alg\) and \(\epsi \geq 0\), it holds that
            \begin{equation*}
                \norm{\commutator{A,B}}
                \leq
                \epsi \, \norm{A}
                \qquad\text{for all \(A\in\alg_{\Lambda}\) and \(\Lambda\in \powerset(\Gamma\setminus X)\)}
                ,
            \end{equation*}
            then
            \begin{equation*}
                \norm{B-\cexpsym_{X} \, B}
                \leq
                2 \, \epsi
                .
            \end{equation*}
    \end{thmlist}
\end{lemma}

Let \((s,t)\in \Delta^2_I\), \(Y\in \powerset(\Gamma)\), \(B\in \alg_Y\) and \(r\geq 0\).
We set \(B_r \coloneq \cexpsym_{Y[r]} \, \tau_{s,t} \, B\).
For \(Y=\emptyset\), the only possible operator is \(B=c \, 1\) and therefore \(\tau_{s,t} \, B = B = B_r\) for all~\(r\).
We now restrict to \(Y \neq \emptyset\).
To apply \cref{lem:conditional-expectation-spins}, let \(\Lambda\in \powerset(\Gamma \setminus Y[r])\) and \(A \in \alg_\Lambda\).
Then, there exists \(R \geq r\) such that \(\Lambda\subseteq Y[R']\) for all \(R' \geq R\) and \(A\in \alg_{Y[R']\setminus Y[r]}\).
By \cref{thm:Lieb-Robinson-bound} with \(K=\commutator{A,\cdot}\), the evolutions \(\tau^{R'}_{s,t} \, B\) generated by the restricted interactions \(\Phi\rvert_{Y[R']}\) satisfy
\begin{equation}
    \label{eq:proof-localization-LRB}
    \norm{
        \commutator{A, \tau^{R'}_{s,t} \, B}
    }
    \leq
    2 \, \norm{A} \, \norm{B}
    \, C_F^{-1}
    \, \paren[\big]{\e^{C_F \, I_{\Phi^{R'}}(s,t)}-1}
    \, \sum_{x\in Y} \, \sumstack[r]{y\in Y[R']\setminus Y[r]} \, F\pdist{x,y}
    ,
\end{equation}
with
\begin{equation*}
    \Phi^{R'}(Z,t)
    =
    \begin{cases*}
        0
         & if \(Z\subseteq Y\) or \(Z\subseteq Y[R']\setminus Y[r]\) or \(\abs{Z}=1\),
        \\
        \Phi\rvert_{Y[R']}(Z,t)
         & otherwise.
    \end{cases*}
\end{equation*}
Notice that \(\norm{\Phi^{R'}(u)}_F \leq \norm{\tilde{\Phi}(u)}_F\) for all \(u\in I\), which lets us replace \(I_{\Phi^{R'}}\) with \(I_{\tilde{\Phi}}\) in~\eqref{eq:proof-localization-LRB}.
Enlarging the second sum to \(y \in \Gamma \setminus Y[r]\) provides a uniform upper bound.
Therefore, by the convergence of \(\tau^{R'}_{s,t} \, B\) in the limit \(R'\to\infty\) (see \cref{prop:existence-of-the-thermodynamic-limit}) we obtain
\begin{equation*}
    \norm{
        \commutator{A, \tau_{s,t} \, B}
    }
    \leq
    2 \, \norm{A} \, \norm{B}
    \, C_F^{-1}
    \, \paren[\big]{\e^{C_F \, I_{\tilde{\Phi}}(s,t)}-1}
    \, \sum_{x\in Y} \, \sumstack[r]{y\in \Gamma\setminus Y[r]} \, F\pdist{x,y}
    .
\end{equation*}
Now, we can apply \cref{lem:conditional-expectation-spins} to bound \(\norm{\tau_{s,t} \, B - \cexpsym_{Y[r]} \, \tau_{s,t} \, B}\) and thus conclude the statement.

\appendix
\crefalias{section}{appendix}

\section{Proof of \texorpdfstring{\cref{prop:finite-evolutions}}{Proposition~\ref*{prop:finite-evolutions}}}
\label{app:proof-time-evolution-properties}

\finiteevolutions*
\begin{proof}[Proof of \cref{prop:finite-evolutions}]
    For the proof, let
    \begin{equation*}
        K(t)
        \coloneq
        \sumstack[lr]{Z\in \powerset(\Gamma)} \, \Phi(Z,t)
        .
    \end{equation*}
    Since \(\Phi\) is finitely supported and each map \(t\mapsto\Phi(Z,t)\) is norm-continuous and uniformly bounded, the same holds for \(t\mapsto K(t)\).
    Moreover, \(K(t)\) is dissipative for all~\(t\in I\), since the sum of two bounded dissipative operators is again dissipative, by Lumer-Phillips and the Trotter product formula.
    Then, for all \((s,t)\in \Delta^2_I\) we have
    \begin{equation*}
        \tau_{s,t}
        =
        \sum_{n=0}^\infty
        \intlap_{\Delta^n_{s,t}} \odif{u}
        \, K(u_1) \, \dotsm \, K(u_n)
        ,
    \end{equation*}
    which is bounded by
    \begin{equation}
        \label{eq:proof-prop-finite-evolutions-series-bound}
        \norm{\tau_{s,t}}
        \leq
        \sum_{n=0}^\infty
        \frac{\abs{t-s}^n}{n!}
        \, \paren[\Big]{
            \;\sup_{u\in \intervalcc{s,t}} \, \norm{K(u)}
        }^n
        .
    \end{equation}

    Statement~\itemref{item:prop-finite-evolutions:cocycle} can be shown by writing out the series on both sides of the equation and using that for \(N\in \N_0\) and \(t \in \intervalcc{s,u}\) the simplex \(\Delta_{s,u}^N\) can be decomposed into smaller simplices with sizes \(\abs{t-s}\) and \(\abs{u-t}\) as
    \begin{equation*}
        \Delta_{s,u}^N
        =
        \bigcup_{n=0}^N \Delta_{s,t}^n \times \Delta_{t,u}^{N-n}
        ,
    \end{equation*}
    with the union being disjoint up to sets of measure zero.
    Consequently,
    \begin{align*}
        \Alignindent
        \intlap_{\Delta^N_{s,u}} \odif{v}
        \, K(v_1)  \dotsm  K(v_N)
        =
        \sum_{n=0}^N
        \paren[\bigg]{
            \intlap_{\Delta^n_{s,t}} \odif{v}
            \, K(v_1)  \dotsm  K(v_n)
        }
        \,
        \paren[\bigg]{
            \intlap_{\Delta^{N-n}_{t,u}} \odif{w}
            \, K(w_1)  \dotsm  K(w_{N-n})
        }
        ,
    \end{align*}
    which, together with a change in summation index, implies the statement. 
    Statement~\itemref{item:prop-finite-evolutions:identity} follows immediately from the definition.
    For~\itemref{item:prop-finite-evolutions:continuity}, observe that the integrand
    \begin{equation*}
        K(u_1) \, \dotsm \, K(u_n)
    \end{equation*}
    is continuous in \(u_1,\dotsc,u_n\).
    Hence, the integral
    \begin{equation*}
        \intlap_{\Delta^n_{s,t}} \odif{u}
        \, K(u_1) \, \dotsm \, K(u_n)
    \end{equation*}
    is also continuous in both~\(s\) and~\(t\).
    With the bound~\eqref{eq:proof-prop-finite-evolutions-series-bound} and dominated convergence, one can show that \(\tau_{s,t}\) is also continuous in \(s\) and \(t\).
    For statement~\itemref{item:prop-finite-evolutions:derivative-second}, we first observe that we can reorder the integrals into
    \begin{equation*}
        \intlap_{\Delta^n_{s,t}} \odif{u}
        \, K(u_1) \, \dotsm \, K(u_n)
        =
        \intlap_s^t \odif{u_1} \intlap_s^{u_1} \odif{u_2} \, \dotsi \intlap_{s}^{u_{n-1}} \odif{u_n} \, K(u_n) \, \dotsm \, K(u_1)
        .
    \end{equation*}
    This simplifies taking the derivative
    \begin{align*}
        \Alignindent
        \sum_{n=0}^\infty \partial_t \intlap_s^t \odif{u_1} \intlap_s^{u_1} \odif{u_2} \, \dotsi \intlap_{s}^{u_{n-1}} \odif{u_n} \, K(u_n) \, \dotsm \, K(u_1)
        \\&=
        \sum_{n=1}^\infty \intlap_s^{t} \odif{u_2} \, \dotsi \intlap_{s}^{u_{n-1}} \odif{u_n} \, K(u_n) \, \dotsm \, K(u_2) \, K(t)
        \\&=
        \tau_{s,t} \, K(t)
        .
    \end{align*}
    It only remains to exchange the sum and derivative, which is justified because for every compact interval \(J=\intervalcc{s,b} \subseteq I\) one has \(\partial_t \, \id_{\alg} = 0\) and
    \begin{equation*}
        \sup_{t\in J}
        \, \norm[\bigg]{
            \partial_t \intlap_s^t \odif{u_1} \intlap_s^{u_1} \odif{u_2} \, \dotsi \intlap_{s}^{u_{n-1}} \odif{u_n} \, K(u_n) \, \dotsm \, K(u_1)
        }
        \leq
        \sup_{t\in J} \, \frac{\abs{t-s}^{n-1}}{(n-1)!} \, \paren[\Big]{
            \, \sup_{u\in J} \norm{K(u)}
        }^n
    \end{equation*}
    for \(n \geq 1\), and this bound is summable.
    Statement~\itemref{item:prop-finite-evolutions:derivative-first} follows in exactly the same way without reordering the integrals.
    To prove the final statement~\itemref{item:prop-finite-evolutions:contractive}, set \(u^l_k \coloneq s + k \, \frac{t-s}{l}\) for \(l\in \N_{>0}\) and \(k\in \List{0,1,\dotsc,l}\) and approximate the map \(\intervalcc{s,t} \to \calL(\alg), \, u \mapsto K(u)\) by a sequence of piecewise constant functions \((K^l)_{l\in \N_{>0}}\)
    \begin{equation*}
        K^l\colon \intervalcc{s,t} \to \calL(\alg)
        ,\
        u \mapsto K\paren[\big]{\max \, \Set{u^l_k \given u^l_k \leq u}}
        .
    \end{equation*}
    For \((u,v)\in \Delta^2_{s,t}\) we define the time-ordered exponential of \(K^l\) as
    \begin{equation*}
        \tau^l_{u,v}
        \coloneq
        \sum_{n=0}^\infty
        \intlap_{\Delta^n_{u,v}} \odif{w}
        \, K^l(w_1) \, \dotsm \, K^l(w_n)
        .
    \end{equation*}
    We note that by the same argument used to prove statement~\itemref{item:prop-finite-evolutions:cocycle} it holds for all \((u,v,w) \in \Delta^3_{s,t}\) that \(\tau^l_{u,v} \, \tau^l_{v,w} = \tau^l_{u,w}\).
    Moreover, by continuity of \(K\) and dominated convergence it holds that \(\tau^l_{u,v} \to \tau_{u,v}\) as \(l\to \infty\).
    Hence, we obtain
    \begin{equation*}
        \norm{\tau_{s,t}}
        =
        \lim_{l\to \infty} \norm{\tau^l_{s,t}}
        \leq
        \limsup_{l\to \infty} \, \norm{\tau^l_{u^l_0, u^l_1}} \, \norm{\tau^l_{u^l_1, u^l_2}} \, \dotsm \, \norm{\tau^l_{u^l_{l-1}, u^l_l}}
        .
    \end{equation*}
    Since the function \(K^l\) is constant on the intervals \(\intervalco{u^l_k,u^l_{k+1}}\) with \(k\in \List{0,\dotsc,l-1}\), the time-ordered exponential \(\tau^l_{u^l_k,u^l_{k+1}}\) reduces to an ordinary exponential
    \begin{equation*}
        \tau^l_{u^l_k,u^l_{k+1}}
        =
        \Exp[\Big]{\frac{t-s}{l} \, K(u^l_k)}
        .
    \end{equation*}
    As we have observed previously, \(K(u)\) is a bounded dissipative operator for all \(u\in I\).
    Then by the Lumer-Phillips Theorem, it generates a contractive semigroup and the norm of the above exponential is bounded by \(1\).
    Therefore, we can conclude
    \begin{equation*}
        \norm{\tau_{s,t}}
        \leq
        \limsup_{l\to \infty}
        \, \norm[\Big]{\Exp[\Big]{\frac{t-s}{l} \, K(u^l_0)}}
        \, \dotsm
        \, \norm[\Big]{\Exp[\Big]{\frac{t-s}{l} \, K(u^l_{l-1})}}
        \leq
        1
        .
        \qedhere
    \end{equation*}
\end{proof}

\statement{Acknowledgments}
This project was supported by the \foreignlanguage{ngerman}{Deutsche Forschungsgemeinschaft} (DFG, German Research Foundation) –
465199066; 
470903074
.
The research of T.W.~was also partially supported by the \foreignlanguage{ngerman}{Deutsche Forschungsgemeinschaft} (DFG, German Research Foundation) under Germany’s Excellence Strategy – EXC-2111 – 390814868.

\statement{AI declaration}
The proof of the result and the original draft were produced in early 2026 without the use of Large Language Models (LLMs).
LLMs were used to check the draft for errors, which led to various minor improvements and corrections.
The authors take full responsibility for the manuscript in its entirety.

\printbibliography[heading=bibintoc]

\end{document}